\documentclass[12pt]{article}
\usepackage[utf8]{inputenc}
\usepackage[a4paper]{geometry}
\usepackage{color}
\usepackage{textcomp}
\usepackage{verbatim}
\usepackage{amsmath}
\usepackage{amssymb} % for mathbb and more
\usepackage{graphicx}
\usepackage{authblk}

\usepackage[T1]{fontenc}
\usepackage{amsthm}

\usepackage{mathrsfs} % for mathscr
\usepackage{enumerate}

\makeatletter

\providecommand{\tabularnewline}{\\}

\usepackage{fancyhdr}% This should be set AFTER setting up the page geometry
\usepackage{amsthm}\usepackage{caption}\usepackage{subcaption}\usepackage{comment}\usepackage{appendix}

\theoremstyle{plain} \newtheorem{theorem}{Theorem}[section] % plain for italic theorems
\theoremstyle{definition} 
\theoremstyle{plain} \newtheorem{lemma}{Lemma}
\theoremstyle{plain}
\theoremstyle{plain}\newtheorem{prop}{Proposition}
\makeatother

\begin{document}

\title{Trapped modes in zigzag graphene nanoribbons}

%\author{%%%%%%%%%%%%%%%%%%%%%

\author[1]{V. A. Kozlov\thanks{vladimir.kozlov@liu.se}}
\author[2]{S. A. Nazarov\thanks{srgnazarov@yahoo.co.uk}} 
\author[1]{A. Orlof\thanks{anna.orlof@liu.se}}

\affil[1]{\textit {Mathematics and Applied Mathematics, MAI, Link\"{o}ping University, SE-58183 Link\"{o}ping, Sweden}}
\affil[2]{\textit{ Saint-Petersburg State University, Universitetsky pr., 28, Peterhof, 198504, St. Petersburg, Russia,
Peter the Great St. Petersburg Polytechnic University, St. Petersburg 195251, Russia,
 Institute of Problems of Mechanical Engineering
RAS, V.O., Bolshoi pr., 61, 199178, St.-Petersburg, Russia}}

\maketitle
\begin{abstract}
We study a scattering on an ultra-low potential in zigzag graphene
nanoribbon. Using mathematical framework based on the continuous Dirac
model and augumented scattering matrix, we derive a condition for
the existence of a trapped mode. We consider the threshold energies where the continuous spectrum changes its multiplicity and show that the trapped modes may appear for energies slightly less than a threshold and its multiplicity does not exceeds one. We prove that trapped modes do not appear outside the threshold, provided the potential is sufficiently small.
\end{abstract}

\section{Introduction}

The problem of disorder in graphene nanoribbons has been studied extensively.
The main purpose of those studies is to eliminate disorder completely
and produce pure high-quality graphene nanoribbons \cite{Chen}. Approaching the
goal of graphene nanoribbons free of impurities and other defects
one can focus on production of such deliberately for the use in electronic
devices. One of desired features for graphene to possess is a possibility
of electron localization. Such localization is difficult to achieve
due to Klein tunneling \cite{Kat}. As graphene electrons behave like massless
particles they undergo tunneling through barriers. However, due to
interference between continuous states of the nanoribbon and a localized
state of a disorder, a trapped mode can be produced.
There are several types of disorder including short-range and long-range. Impurities such as vacancies and adatoms are classified as short-range type and can be described
by a sharp potential, that varies on the scale shorter than graphene
lattice constant (0.246 nm) \cite{Ando}. On the other hand, electric
or magnetic field, interactions with substrate, Coulomb charges \cite{Libisch}, ripples and wrinkles can lead to long-range disorder described by smooth potential (a Gaussian
for example). In the present studies we assume that graphene is free
of short-range defects and the potential is modeled as a long range
one. 

There are two groups of graphene nanoribbons, that differ by the edge
type and are called zigzag and armchair \cite{Brey,Castro}. In this paper, we give a
condition for existence and a choice of ultra-low potential that produces
a trapped mode in zigzag graphene nanoribbon. We work within continuous
Dirac model, where graphene is isotropic and its electrons dynamics
can be described by a system of 4 equations \cite{Castro}

\begin{equation}
\left(\begin{array}{cccc}
0 & i\partial_{x}+\partial_{y} & 0 & 0\\
i\partial_{x}-\partial_{y} & 0 & 0 & 0\\
0 & 0 & 0 & -i\partial_{x}+\partial_{y}\\
0 & 0 & -i\partial_{x}-\partial_{y} & 0
\end{array}\right)+\delta\mathcal{P}\left(\begin{array}{c}
u'\\
v'\\
u\\
v
\end{array}\right)=\omega\left(\begin{array}{c}
u'\\
v'\\
u\\
v
\end{array}\right),\label{eq:Diractot}
\end{equation}
where $\omega=\frac{E}{\hbar\nu_{F}}$ is scaled energy (with $E$ denoting energy and 
$\nu_{F}\approx10^{6}\frac{m}{s}$ Fermi velocity), a potential
$\mathcal{P}$ is a real-value function with compact support such
that $\sup|\mathcal{P}|\leq1$ and $\delta$ is a real-value small
parameter.

The number of equation is a consequence of the discrete description of
graphene lattice and low energy approximation which leads to the continuous
model \cite{Castro, Macucci}. In the discrete model graphene is described as a composition of two triangular interpenetrating lattices of carbon atoms (called
A and B) \cite{Castro}. Then low energy approximation can be done
in a twofold way, close to two different energy minima (called $K$ and
$K^\prime$) in the graphene dispersion relation. Consequently, in the continuous
model, we have two waves that describe an electron in any single point
of the ribbon (A or B) coupled in two different ways, close to $K$ (first
two equations in (\ref{eq:Diractot})) or $K^\prime$ point (last two equations
in (\ref{eq:Diractot})). A nanoribbon is modeled as an unit strip $\Pi=(0,1)\times\mathbb{R}$ due to rescaling.
The zigzag boundary of the nanoribbon requires one wave (A) to disappear
specifically at one edge and the other (B) at the other one
\cite{Brey}
\begin{equation}
u'(0,y)=0\ ,\ \ u(0,y)=0\ ,\ \ v'(1,y)=0\ ,\ \ v(1,y)=0\;.\label{BCDiractot}
\end{equation}
As our potential $\mathcal{P}$ is assumed to be of long-range type, it
can be described by a diagonal matrix with equal elements \cite{Ando}.
As neither the potential nor the boundary conditions couples $K$ and
$K^\prime$ valleys, the system of 4 equations can be split into two systems
of 2 equations where only intravalley scattering is allowed. We consider
one of them (two last equations in (\ref{eq:Diractot}))
\begin{equation}
\left(\begin{array}{cc}
0 & -i\partial_{x}+\partial_{y}\\
-i\partial_{x}-\partial_{y} & 0
\end{array}\right)\left(\begin{array}{c}
u\\
v
\end{array}\right)+\delta\mathcal{P}\left(\begin{array}{c}
u\\
v
\end{array}\right)=\omega\left(\begin{array}{c}
u\\
v
\end{array}\right)\label{homoDirac2P}
\end{equation}
supplied with the boundary conditions:\\
\begin{equation}
u(0,y)=0\ ,\ \ v(1,y)=0.\label{homoDirac3P}
\end{equation}

\begin{figure}[h]
\noindent \centering{}\includegraphics[scale=2]{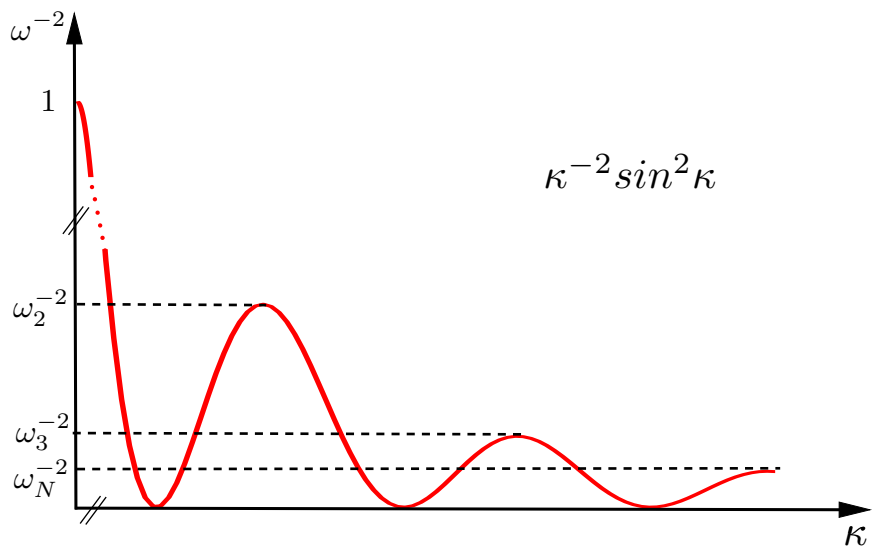}\caption{Dispersion with energy thresholds. \label{fig:Plot_sin_intro}}
\end{figure}

Energy threshold $\omega=\omega_{N}>1$, $N=2,\,3,\dots$ is defined
by one of maxima in zigzag dispersion relation $\omega^{-2}=\kappa^{-2}\sin^{2}\kappa$
and it reads $\frac{d}{d\kappa}\Big(\kappa^{-2}\sin^{2}\kappa\Big)=0$
(see Figure \ref{fig:Plot_sin_intro}). The index $N$ defines a threshold energy
$\omega_{N}$ as it indicates change of the multiplicity of the continuous
spectrum from $2N-3$ to $2N-1$ for $N\geq2$. A trapped mode is defined as a vector eigenfunction (from $L_2$ space) that corresponds to an eigenvalue embedded in the continuous spectrum. 
The main result of the paper is the following theorem about existence of trapped modes
in zigzag graphene nanoribbon for energies close to one of the thresholds that can be chosen arbitrary.

\begin{theorem}\label{Texistence} For every $N=2,\,,3,\,...$ there exists $\varepsilon_N>0$ such that for each
$\varepsilon\in(0,\varepsilon_{N})$ there exists $\delta\sim\sqrt{\varepsilon}$
and a potential $\mathcal{P}$ such that the problem {\rm(\ref{homoDirac2P})},
{\rm(\ref{homoDirac3P})} has a trapped mode solution for $\omega^{-1}=\omega_N^{-1}+\varepsilon$.

\end{theorem}

Second result shows that trapped modes may appear only for energy
slightly smaller than threshold and that spectrum far from the
threshold is free of embedded eigenvalues. Moreover their multiplicity does
not exceed one. 

\begin{theorem}\label{Tmultiplicity} There exist positive numbers $\varepsilon_{0}$ and $\delta_{0}$, which
may depend on $N$, such that if

{\rm(i)} $\omega\in[\omega_{N},\,\omega_{N}+\varepsilon_{0}]$
and $|\delta|<\delta_{0}$, then problem {\rm(\ref{homoDirac2P})}, {\rm(\ref{homoDirac3P})}
has no trapped modes;

{\rm(ii)} $\omega\in[\omega_{N}-\varepsilon_{0},\,\omega_{N})$
and $|\delta|<\delta_{0}$, then the multiplicity of a trapped mode to problem {\rm(\ref{homoDirac2P}), (\ref{homoDirac3P})} does not exceed $1$.

{\rm(iii)} For every $\varepsilon_1>0$ and $C_1>0$ there exist $\delta_1>0$ such that if $0\leq\omega< C_1$ and $|\omega-\omega_N|>\varepsilon_1$ for all $N=2,\ldots$ satisfying $\omega_N\leq C_1$ and $|\delta|<\delta_1$, then problem {\rm(\ref{homoDirac2P}), (\ref{homoDirac3P})} has no trapped modes.
\end{theorem}

Those results come from the analysis of the trapped modes in the $K^{\prime}$ valley (system (\ref{homoDirac2P}), (\ref{homoDirac3P})), however the analysis in the $K$ valley (first two equations in (\ref{eq:Diractot}) with boundary conditions (\ref{BCDiractot})) is analogous and requires complex conjugation only.

The continuous spectrum of the problem (\ref{homoDirac2P}), (\ref{homoDirac3P})
with $\mathcal{P}=0$ covers the whole real axis and, hence, its
eigenvalues, if exist, possess the natural instability, that is, a
small perturbation may lead them out from the spectrum and turn into
points of complex resonance, cf.~\cite{AsPaVa,na510} and the review
paper \cite{LM}. A few of approaches have been proposed to compensate
for this instability and to detect eigenvalues embedded into the continuous
spectrum. First of all, a simple but very elegant trick was developed
in \cite{EvLeVa} for scalar problems. Namely, under a symmetry
assumption on waveguide's shape an artificial Dirichlet condition
is imposed on the mid-hyperplane of the waveguide which shifts the
lower bound of the spectrum above and allows to apply the variational
or asymptotic method to find out a point in the discrete spectrum
of the reduced problem. At the same time, the odd extension of the
corresponding eigenfunciton through the Dirichlet hyperplane gives
an eigenfunction of the original problem so that it remains to verify
that the eigenvalue falls into the original continuous spectrum. In
other words, the problem operator is restricted into a subspace where
it may get the discrete spectrum which becomes a part of the point
spectrum in the complete setting. For vectorial problems the existence
of such invariant subspaces usually demand very strong conditions
on physical and geometrical properties of waveguides and therefore
the trick works rather rarely or needs supplementary ideas, cf.~\cite{Vas2,Vas3}
and \cite{na398}. Unfortunately, the Dirac equations do not possess
the necessary properties and we are not able to find a way to apply
the trick in our problem. 

Another approach accepting formally self-adjoint elliptic systems
but employing much more elaborated asymptotic analysis is based on
the concept of enforced stability of embedded eigenvalues \cite{TMP,na510,FAA}.
In this way, having an eigenvalue in the continuous spectrum of a
waveguide with $N$ open channels for wave propagation one can select
a small perturbation of the problem by means of tuning $N$ parameters
such that, although the eigenvalue enjoys a perturbation, it remains
sitting on the real axis and does not move into the complex plane.
It is remarkable that, as it was shown in a different situations \cite{TMP,na500,na510,na521}
and others, it is possible to take as an "initiator" of a trapped
mode a particular standing wave at the threshold value of the spectral
parameter and by an appropriate choice of the perturbation parameters
to construct an eigenvalue which is situated near but only on one
side of the threshold so that it belongs to the continuous spectrum.
This method was introduced and developed in \cite{TMP,na510,FAA}.
Aiming to apply it for detection of eigenvalues for the zigzag
graphene nanoribbon, we unpredictably observed that the corresponding
boundary value problem in whole is not elliptic (see Appendix \ref{sec:Appendix:-Ellipticity}).
As a result, many steps of the detection procedure require serious
modifications.

The paper is organized as follows. In Sect.~\ref{sec:DE} we analyse Dirac model
without potential. For each non-zero energy we construct all bounded solutions and identify thresholds where the dimension of the space of such solution changes. We construct also unbounded solutions near threshold and introduce a symplectic form, which will play an important role in the study of the scattering problem. These unbounded solutions are studied in Sect.~\ref{sub:solexp}, \ref{sub:BCl} and \ref{sub:PC}. In Sect.~\ref{sub:NH}, we present a solvability result for non-homogenous problem. In Sect.~\ref{sub:Perturbated-Dirac-equation} we add potential to the model, and consider scattering problem with the use of artificial augumented
scattering matrix introduced in Sect.~\ref{sub:S}. In Sect.~\ref{sec:TrappedModes} we analyze trapped modes, providing in Sect.~\ref{sub:NS} a necessary and sufficient condition for their existence, from which in Sect.~\ref{sub:PT1} we extract the potential description and prove Theorem \ref{Texistence}. Finally, in the last section, Sect.~\ref{sec:mult}, we analyze the multiplicity of trapped modes proving Theorem \ref{Tmultiplicity}.

\section{Dirac equation}\label{sec:DE}

\subsection{Problem statement}\label{sec:PS}

We consider problem (\ref{homoDirac2P}) without potential ($\mathcal{P}=0$)
\begin{equation}
{\mathcal D}\left(\begin{array}{c}
u\\
v
\end{array}\right)=
\omega
\left(\begin{array}{c}
u\\
v
\end{array}\right),\;\;\;
{\mathcal D}={\mathcal D}(\partial_x,\partial_y):=
\left(\begin{array}{cc}
0 & -i\partial_{x}+\partial_{y}\\
-i\partial_{x}-\partial_{y} & 0
\end{array}\right)
\label{homoDirac2}
\end{equation}
supplied with the
boundary conditions (\ref{homoDirac3P}).
Our goal is to find solutions to this problem, especially
bounded ones and thereby describe continuous spectrum of the operator
corresponding to (\ref{homoDirac2}), (\ref{homoDirac3P}).

Let us introduce the spaces 
\[
X=\{u\in L^{2}(\Pi)\,:\,(i\partial_{x}+\partial_{y})u\in L^{2}(\Pi),\; \mbox{and}\;\; u(0,y)=0\}
\]
and 
\[
Y=\{v\in L^{2}(\Pi)\,:\,(-i\partial_{x}+\partial_{y})v\in L^{2}(\Pi),\;\mbox{and}\;\; v(1,y)=0\}.
\]
Then ${\mathcal D}$ is a self-adjoint operator in $L^{2}(\Pi)\times L^{2}(\Pi)$
with the domain $X\times Y$.

We note that for $\omega=0$ we have $(-i\partial_{x}+\partial_{y})v=(-i\partial_{x}-\partial_{y})u=0$
therefore $u=u(-x+iy)$ and $v=v(x+iy)$ what together with $u(0,y)=0\ ,v(1,y)=0$
give $u=0\ ,v=0$. There are no non-trivial solutions to (\ref{homoDirac2}),
(\ref{homoDirac3P}) for $\omega=0$.

Now, assume that $\omega\neq0$, then problem (\ref{homoDirac2}), (\ref{homoDirac3P})
can be written as the system
\begin{equation}
-\Delta u=\omega^{2}u\ ,\ \ v=\frac{1}{\omega}(-i\partial_{x}-\partial_{y})u\ ,\ \ u(0,y)=0\ ,\ \ v(1,y)=0.\label{delt}
\end{equation}
We are looking for non-trivial solutions which are exponential (or
possibly power exponential) in $y$, i.e.
\begin{equation}
(u(x,y),v(x,y))=e^{-i\lambda y}(\mathcal{U}(x),\mathcal{V}(x)),\label{uosc}
\end{equation}
$\lambda$ is a component of a wave vector parallel with the nanoribbons edge.
Then insertion into (\ref{delt}) gives 
\begin{equation}
\begin{cases}
-\mathcal{U}_{xx}=(\omega^{2}-\lambda^{2})\mathcal{U}\ ,\ \ \mathcal{U}(0)=0,\ \ \ \ \mathcal{U}_{x}(1)=\lambda\mathcal{U}(1),\\
\mathcal{V}=\frac{1}{\omega}(-i\mathcal{U}_{x}+i\lambda\mathcal{U}).
\end{cases}\label{eq:sep}
\end{equation}

\begin{lemma} 

If {\rm(\ref{uosc})} is a non-trivial solution to {\rm(\ref{homoDirac2})},
{\rm(\ref{homoDirac3P})} with a certain complex $\lambda$ then: {\rm(i)} $\Im\lambda=0$
or {\rm(ii)} $\Im\lambda\neq0$ and $\Re\lambda>0$.

\end{lemma}

\begin{proof}Multiplying the first equation in (\ref{eq:sep}) by
$\overline{\mathcal{U}}$ and integrating in $x\in(0,1)$,
we have
$$
\int_{0}^{1}(\omega^{2}-\lambda^{2})\mathcal{U}\overline{\mathcal{U}}dx=-\int_{0}^{1}\mathcal{U}_{xx}\overline{\mathcal{U}}dx=\int_{0}^{1}\mathcal{U}_{x}\overline{\mathcal{U}}_{x}dx-\lambda\mathcal{U}(1)\overline{\mathcal{U}}(1)
$$
Taking the imaginary part of this equation, we get
$$
2\Re\lambda\,\Im\lambda\int_{0}^{1}\mathcal{U}\,\overline{\mathcal{U}}dx=\Im\lambda\mathcal{U}(1)\,\overline{\mathcal{U}}(1).
$$
Therefore, if $\Im\lambda\neq0$, then $\Re\lambda$ must be positive provided $\mathcal{U}\neq0$.
\end{proof}Consider the case $\lambda^{2}=\omega^{2}$. Then there
exists an exponential solution only for $\lambda=1$ and it has the
following form
\begin{equation}
\mathcal{U}(x)=x\ ,\ \ \mathcal{V}(x)=i\frac{1}{\omega}(x-1).\label{sol10-1}
\end{equation}
Let $\lambda^{2}\neq\omega^{2}$. Then the solution to (\ref{eq:sep})
is given by 
\begin{equation}
\mathcal{U}(x)=\sin(\kappa x)\ ,\ \ \mathcal{V}(x)=\pm i\sin(\kappa(x-1)),\label{sol-1}
\end{equation}
where $\kappa$ satisfies 
\begin{equation}
\frac{\sin\kappa}{\kappa}=\pm\frac{1}{\omega}\label{cond}
\end{equation}
and $\lambda$ can be evaluated from: 
\begin{equation}
\lambda=\kappa \cot\kappa.\label{LUD2}
\end{equation}
One can verify that $\kappa^{2}+\lambda^{2}=\omega^{2}$. Relations
(\ref{sol-1}) can be written also as
\begin{equation}
(\mathcal{U}(x),V(x))=\Big(\sin(\kappa x),i\omega^{-1}(-\kappa\cos (\kappa x)+\lambda\sin(\kappa x)\Big)\:,\:\:\:\frac{\sin\kappa}{\kappa}=\pm\frac{1}{\omega}.\label{eq:wavesanother}
\end{equation}
\subsection{Symmetries\label{sub:Symmetries}}

One can verify that if $\kappa$ solves (\ref{cond}) then $-\kappa$
is also a solution to (\ref{cond}). Moreover if $(u,v)$ is a solution
to the problem (\ref{homoDirac2}), (\ref{homoDirac3P}), then replacing
$\kappa$ by $-\kappa$ in (\ref{eq:wavesanother}) we obtain linearly
dependent solution $-(u,v)$. Thus, it suffices to take only one value
of $\kappa$ satisfying (\ref{cond}), we assume that $arg(\kappa)\in[0,\pi)$.

In what follows we look only at positive $\omega$. If $\omega$ is
negative then according to the second formula (\ref{eq:wavesanother}),
it can be obtained from the corresponding solution $(u,v)$ with positive
$\omega$ by taking the second component $v$ with minus sign. 

Finally, if $(u,v)$ is a solution, then $(\overline{u}(x,-y),-\overline{v}(x,-y))$
is also a solution together with $(v(1-x,y),-u(1-x,y))$.

\subsection{Solutions of the form (\ref{uosc}) with real wave vector $\lambda$ \label{sub:sol}}

Here we construct all solutions to (\ref{homoDirac2}), (\ref{homoDirac3P})
of the form (\ref{uosc}) with real $\lambda$. According to Sect.
\ref{sub:Symmetries}, it is sufficient to consider $\omega>0$ in
(\ref{homoDirac2}). Let us divide the analysis in three cases: $0<\omega<1$,
$\omega=1$ and $\omega>1$.

\begin{figure}[h]
\centering{}\includegraphics[scale=2]{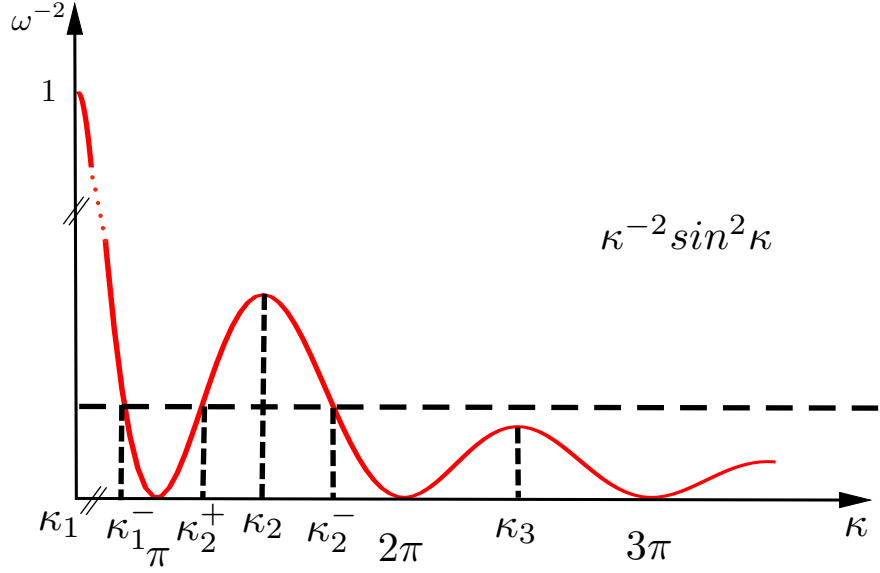}\caption{Dependence of $\kappa$ on $\omega$. \label{fig:Plot-of-sink}}
\end{figure}

\begin{figure}[h]
\centering{}\includegraphics[scale=2]{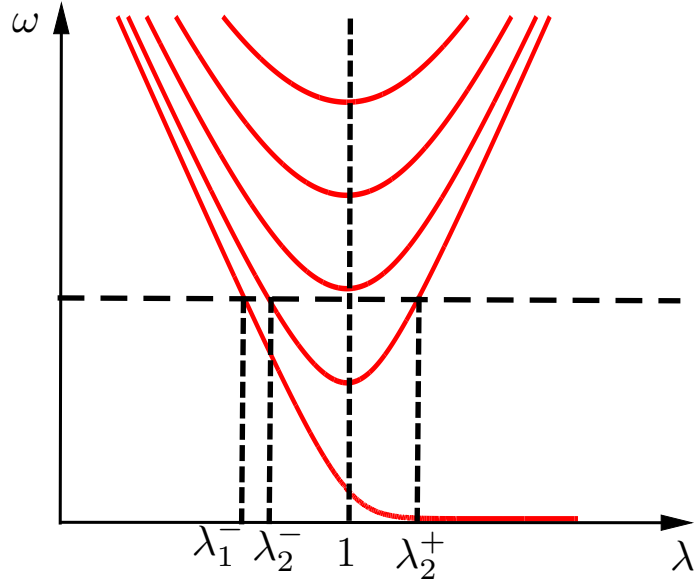}\caption{Energy bands for zigzag graphene nanoribbon. Dependence of wave vector $\lambda$
on energy $\omega$.\label{fig:disp}}
\end{figure}

(1) $0<\omega<1$. Equation (\ref{cond}) has a solution $\kappa=i\tau$
where $\tau$ is real and satisfies 
$$
\frac{\sinh\tau}{\tau}=\pm\frac{1}{\omega}.
$$
Then $\lambda=\tau\coth\tau$ and 
$$
\mathcal{U}(x)=i\sinh(\tau x)\ ,\ \ \mathcal{V}(x)=\mp\sinh(\tau(x-1))
$$

(2) $\omega=1$. The solution to (\ref{cond}) is $\kappa=0$, $\lambda=1$
and the vector $(U,\mathcal{V})$ is given by (\ref{sol10-1}). 

(3) $\omega>1$. Then $\kappa$ is real and satisfies (\ref{cond}) and corresponding $\lambda$
is evaluated by (\ref{LUD2}). In order to describe solutions of (\ref{cond}) the real numbers $\kappa_{j}$ are introduced as the
maximum of $\kappa^{-2}sin^{2}\kappa$ on the interval $[(j-1)\pi,\, j\pi)$,
$j=1,2,\ldots$ (see Figure \ref{fig:Plot-of-sink}). We put 
$$
\frac{1}{\omega_{j}^{2}}=\frac{\sin^{2}\kappa_{j}}{\kappa_{j}^{2}},\;\;\text{and note that}\qquad\frac{d}{d\kappa}\Big(\frac{\sin^{2}\kappa}{\kappa^{2}}\Big)\Big|_{\kappa=\kappa_{j}}=0.\label{eq:max}
$$
Then $\lambda_{j}=1\:\text{and}\:\kappa_{j}=\sqrt{\omega_{j}^{2}-1}.$ From (\ref{cond})
it follows that $\omega_j$ satisfies
\begin{equation}
\frac{\omega^{2}-1}{\omega^{2}}=\sin^{2}(\sqrt{\omega^{2}-1}).\label{K3}
\end{equation}
One can verify that $\kappa_{j}<(2j-1)\pi/2$ and $\kappa_{1}=0$,
$\omega_{1}=1$.
(a) If $\omega\in(\omega_{N-1},\omega_{N})$, $N=2,\,3\ldots$, then
there are $2N-3$ solutions with real $\lambda$, which can be labelled
as follows (see Figure \ref{fig:Plot-of-sink} and Figure \ref{fig:disp})
$$
\lambda_{j}^{\pm}=\kappa_{j}^{\pm}\cot(\kappa_{j}^{\pm}),\:\:\: j=2,\,3,\ldots,\, N-1,
$$
where $\kappa_{j}^{\pm}\in((j-1)\pi,j\pi)$ satisfies 
$$
\frac{\sin\kappa_{j}^{\pm}}{(\kappa_{j}^{\pm})}=\frac{(-1)^{j+1}}{\omega}\;\;\mbox{and \ensuremath{\kappa_{j}^{+}<\kappa_{j}<\kappa_{j}^{-}}}.\label{eq:kappapm}
$$
The corresponding solutions to (\ref{homoDirac2}) are given by 
\begin{equation}
w_{j}^{\pm}(x,y)=e^{-i\lambda_{j}^{\pm}y}(\mathcal{U}_{j}^{\pm}(x),\mathcal{V}_{j}^{\pm}(x)),\label{15Mars1a}
\end{equation}
where 
\begin{equation}
(\mathcal{U}_{j}^{\pm}(x),\mathcal{V}_{j}^{\pm}(x))=(\sin(\kappa_{j}^{\pm}x),(-1)^{j+1}i\sin(\kappa_{j}^{\pm}(x-1))).\label{Uosc2-1}
\end{equation}
There is only one solution of (\ref{cond}) labeled by a negative
index $"-"$ on the interval $(0,\pi)$, which we denote by $\kappa_{1}^{-}$
and corresponding value of $\lambda$ by $\lambda_{1}^{-}$. The corresponding
solution $(\mathcal{U},\mathcal{V})$ is given by $w_{1}^{-}(x,y)=e^{-i\lambda_{1}^{-}y}(\mathcal{U}_{1}^{-}(x),\mathcal{V}_{1}^{-}(x))$,
where $\mathcal{U}_{1}^{-}$ and $\mathcal{V}_{1}^{-}$ are evaluated
by (\ref{Uosc2-1}) with $j=1$.
(b) The case $\omega=\omega_{N}$, $N\geq2$, is called the threshold
case. Here we have $2N-3$ solutions, described already in the case
(a). In addition there are two solutions %$(u_{0}^{0},v_{0}^{0})$ and $(u_{0}^{1},v_{0}^{1})$
with $\lambda=1$ and $\kappa=\kappa_{N}$, which have the form 
\begin{equation}
w_{N}^{0}(x,y)=e^{-iy}(\sin(\kappa_{N}x),(-1)^{N+1}i\sin(\kappa_{N}(x-1)))\label{eq:sol01-1}
\end{equation}
and
\begin{align}
w_{N}^{1}(x,y) & =yw_{N}^{0}(x,y)-\kappa_{N}^{-1}e^{-iy}(ix\cos(\kappa_{N}x),(-1)^{N+1}(1-x)\cos(\kappa_{N}(x-1))),\label{eq:sol02}
\end{align}
where the last solution has linear growth in $y$.

Thus for each $\omega\in(0,\infty)$ there exists a bounded solution
to (\ref{homoDirac2}), (\ref{homoDirac3P}) of the form (\ref{15Mars1a}).
Hence the continuous spectrum
of the Dirac operator ${\mathcal D}$ is the whole real line, i.e. 
$$
\sigma_{c}=\mathbb{R}.
$$

\subsection{Solutions of the form (\ref{uosc}) with non-real wave vector $\lambda$\label{sub:solexp}}

We will see in what follows that trapped modes may be generated by
solutions to (\ref{homoDirac2}), (\ref{homoDirac3P}) with non-real
$\lambda$. In this section we describe such solutions with $\lambda$
close to the real axis.

\begin{figure}[h]
\centering{}\includegraphics[scale=0.5]{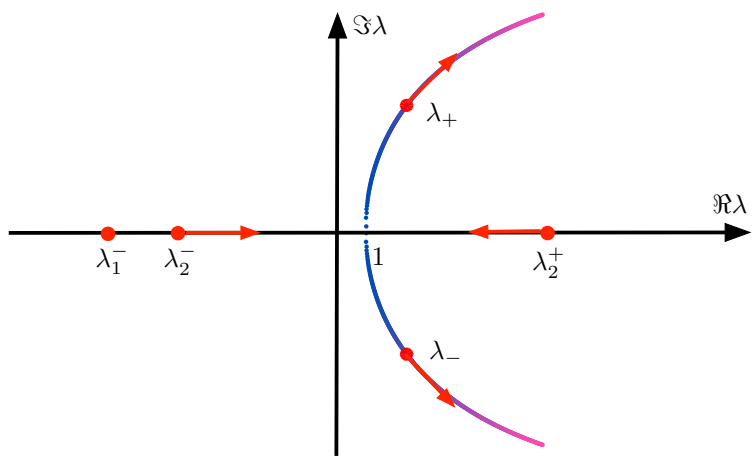}\caption{Bifurcation of $\lambda$ from the threshold. \label{fig:Plot_complex_lambda}}
\end{figure}

\label{Sect2}

Consider the case when $\omega$ is close to $\omega_{N}$, $N=2,\,3,\ldots$
We introduce a small positive parameter $\varepsilon$ and denote by
$\omega_{\varepsilon}$ the energy satisfying $\omega_{\varepsilon}^{-1}=\omega_{N}^{-1}+\varepsilon$.
Then the double root $\lambda=1$ of (\ref{LUD2}) with $\omega=\omega_{N}$
bifurcates into two roots $\lambda_{\pm}=\lambda_{\pm}(\varepsilon)$ (see Figure \ref{fig:Plot_complex_lambda}).
These roots can be found from the equation 
\begin{equation}
g(\lambda,\varepsilon)=\frac{(-1)^{N+1}}{\omega_{\varepsilon}},\;\; g(\lambda,\varepsilon)=\frac{\sin\sqrt{\omega_{\varepsilon}^{2}-\lambda^{2}}}{\sqrt{\omega_{\varepsilon}^{2}-\lambda^{2}}}.\label{15March31a}
\end{equation}
From the definition (\ref{15March31a}), it follows
that $g(\lambda,\varepsilon)$ is an analytic function of $(\omega_{\varepsilon}^{2}-\lambda^{2})$.
Using Taylor's formula for the function $g$ near the point $(\lambda,\varepsilon)=(1,0)$,
we get 
\begin{eqnarray}
 &  & g(\lambda,\varepsilon)-\frac{(-1)^{N+1}}{\omega_{N}}=(-1)^{N}\Big(\varepsilon+\frac{(\lambda-1)^{2}}{2\omega_{N}(\omega_{N}^{2}-1)}+\frac{\omega_{N}^{2}(\lambda-1)\varepsilon}{\omega_{N}^{2}-1}\nonumber \\
 &  & +\frac{(2+3\omega_{N}^{2})(\lambda-1)^{3}}{6\omega_{N}(\omega_{N}^{2}-1)^{2}}\Big)+O(|\lambda-1|^{4}+|\lambda-1|^{2}\varepsilon+\varepsilon^{2}).\label{eq:expansion}
\end{eqnarray}
By means of (\ref{eq:expansion}) one can find the expansions
\begin{equation}
\lambda_{\pm}=1\pm\sqrt{\varepsilon}\lambda_{1}i+\varepsilon\lambda_{2}+O(\varepsilon^{3/2}),\label{eq:lambdapm}
\end{equation}
where $\lambda_{1}=\sqrt{2\omega_{N}(\omega_{N}^{2}-1)}$ and $\lambda_{2}=2\omega_{N}/3$.
Since $\overline{g(\lambda,\varepsilon)}=g(\overline{\lambda},\varepsilon)$,
it follows that $\overline{\lambda_{+}}=\lambda_{-}$. Then, from
the energy relation $\lambda_{\pm}^{2}+\kappa_{\pm}^{2}=\omega_{\varepsilon}^{2}$,
we find 
\begin{equation}
\kappa_{\pm}=\kappa_{N}\mp\sqrt{\varepsilon}\kappa_{N1}i+O(\varepsilon)\label{eq:kappaexp}
\end{equation}
where $\kappa_{N}=\sqrt{\omega_{N}^{2}-1}$ and $\kappa_{N1}=2\sqrt{\omega_{N}}$.
Now, solutions to (\ref{homoDirac2}), (\ref{homoDirac3P}) of the
form (\ref{uosc}) with $\lambda_{\pm}$ are given by
\begin{equation}
w_{N}^{\pm}(x,y)=e^{-i\lambda_{\pm}y}\left(\begin{array}{c}
\sin(\kappa_{\pm}x)\\
(-1)^{N+1}i\sin(\kappa_{\pm}(x-1))
\end{array}\right)\label{exp}
\end{equation}
and by (\ref{eq:lambdapm})
and (\ref{eq:kappaexp}) it can
be written in terms of $w_{0}^{N}$ and $w_{1}^{N}$ (see (\ref{eq:sol01-1}) and
(\ref{eq:sol02})) as
\begin{equation}
w_{N}^{\pm}(x,y)=w_{0}^{N}\pm\sqrt{\varepsilon}\sqrt{2\omega_{N}(\omega_{N}^{2}-1)}w_{1}^{N}+O(\varepsilon).\label{eq:wpmapprox}
\end{equation}
Waves (\ref{exp}) are not analytic in $\varepsilon$.
Consider instead their linear combinations
\begin{align}
\mathbf{w}_{N}^{\varepsilon+}(x,y) & =\frac{1}{2\pi i}\int\limits _{|\lambda-1|=a}\frac{e^{-i\lambda y}}{g(\lambda,\varepsilon)-(-1)^{N+1}\omega_{\varepsilon}^{-1}}\left(\begin{array}{c}
\sin(\kappa(\lambda)x)\\
(-1)^{N+1}i\sin(\kappa(\lambda)(x-1))
\end{array}\right)d\lambda\nonumber \\
 & =\gamma_{+}^{\varepsilon}w_{N}^{+}+\gamma_{-}^{\varepsilon}w_{N}^{-},\label{eq:we+}
\end{align}
and
\begin{align}
\mathbf{w}_{N}^{\varepsilon-}(x,y) & =\frac{1}{2\pi i}\int\limits _{|\lambda-1|=a}\frac{(\lambda-1)e^{-i\lambda y}}{g(\lambda,\varepsilon)-(-1)^{N+1}\omega_{\varepsilon}^{-1}}\left(\begin{array}{c}
\sin(\kappa(\lambda)x)\\
(-1)^{N+1}i\sin(\kappa(\lambda)(x-1))
\end{array}\right)d\lambda\nonumber \\
 & =\gamma_{+}^{\varepsilon}(\lambda_{+}-1)w_{N}^{+}+\gamma_{-}^{\varepsilon}(\lambda_{-}-1)w_{N}^{-},\label{eq:we-}
\end{align}
where $\kappa(\lambda)=\sqrt{\omega^{2}-\lambda^{2}}$ and 
\begin{equation}
\gamma_{\pm}^{\varepsilon}=\Big({\frac{\partial g}{\partial\lambda}(\lambda_{\pm},\varepsilon)}\Big)^{-1},\label{eq:defgamma}
\end{equation}
one can verify that $\gamma_{-}^{\varepsilon}=\overline{\gamma_{+}^{\varepsilon}}$.
Here $a$ is such that the disc $\{|\lambda-1|\leq a\}$ contains exaclty
two solutions $\lambda_{+}$ and $\lambda_{-}$ to (\ref{15March31a}).
Above waves are analytic in $\varepsilon$ because the function $(g(\lambda,\varepsilon)-(-1)^{N+1}\omega_{\varepsilon}^{-1})^{-1}$
is analytic in $\varepsilon$ for $\lambda$ satisfying $|\lambda-1|=a$. 

Now let us consider waves (\ref{eq:we+}) and (\ref{eq:we-})
in the limit case $\varepsilon=0$. First from (\ref{eq:expansion})
we obtain the following expansion $\frac{\partial g}{\partial\lambda}$
near the point $(\lambda,\varepsilon)=(1,0)$

\begin{equation}
\frac{\partial g}{\partial\lambda}(\lambda,\varepsilon)=(-1)^{N}(\frac{\lambda-1}{\omega_{N}(\omega_{N}^{2}-1)}+\frac{\omega_{N}^{2}\varepsilon}{(\omega_{N}^{2}-1)}+\frac{(2+3\omega_{N}^{2})(\lambda-1)^{2}}{2\omega_{N}(\omega_{N}^{2}-1)^{2}}\Big)+O(|\lambda-1|^{3}+|\lambda-1|\varepsilon).\label{eq:dgamapprox}
\end{equation}
From (\ref{eq:defgamma}) and (\ref{eq:dgamapprox}) combined with
(\ref{eq:lambdapm}), we get the following expansion 
\begin{equation}
\gamma_{\pm}^{\varepsilon}=(-1)^{N}\frac{\sqrt{\omega_{N}(\omega_{N}^{2}-1)}}{\sqrt{\varepsilon}\sqrt{2}i}\Big(\pm1+i\sqrt{\varepsilon}
\frac{\sqrt{\omega_N}(-3\omega_N^2+\omega_N-\frac{4}{3})}{\sqrt{2(\omega_N^2-1)}}
\Big)+O(\sqrt{\varepsilon}).\label{eq:gammaapprox}
\end{equation}
Finally, combining (\ref{eq:wpmapprox}) with (\ref{eq:gammaapprox}),
we get

\begin{equation*}
\mathbf{w}_{N}^{\varepsilon+}=\tilde{A}w_{0}^{N}-i\tilde{B}w_{1}^{N}+O(\varepsilon)
\end{equation*}
and
\begin{equation*}
\mathbf{w}_{N}^{\varepsilon-}=\tilde{B}w_{0}^{N}+O(\varepsilon),
\end{equation*}
where $\tilde{A}$ and $\tilde{B}$ are real constants given by 
\begin{equation*}
\tilde{A}=(-1)^{N+1}2\omega_{N}^{2}(\omega_{N}^{2}-1)(2\omega_{N}^{2}+\frac{4}{3}),\:\:\:\tilde{B}=(-1)^{N}2\omega_{N}(\omega_{N}^{2}-1).
\end{equation*}

\subsection{Location of the wave vector $\lambda$}\label{sub:RE}

The forthcoming analysis, which is based on the Laplace transform
of the problem (\ref{homoDirac2}), (\ref{homoDirac3P}) with respect
to $y$, requires a knowledge of location of roots to equation (\ref{cond}), (\ref{LUD2})
or equivalently of the equation 
\begin{equation}
\cos\kappa-\lambda\frac{\sin\kappa}{\kappa}=0,\;\;\kappa^{2}=\omega^{2}-\lambda^{2}.\label{K1}
\end{equation}
Let us denote the left-hand side of (\ref{K1})
by $\mathcal{F}(\lambda)$, which  is analytic
with respect to $\lambda$. One can verify that 
$$
\frac{d}{d\lambda}\mathcal{F}(\lambda)=\frac{\lambda^{2}}{\kappa^{2}}\mathcal{F}(\lambda)+\frac{\omega^{2}(\lambda-1)}{\kappa^{2}}\frac{\sin\kappa}{\kappa},\;\;\;\kappa=\sqrt{\omega^{2}-\lambda^{2}}.
$$
We collect required properties of the roots in the following

\begin{prop}\label{Pr1} {\rm(i)} All roots $\lambda$ of {\rm(\ref{K1})}
are simple except of the case $\omega>1$ and $\omega$ is at the threshold - it is a root of {\rm (\ref{K3})}. In this case the root $\lambda=1$ is double and all the other roots
are simple.

{\rm(ii)} Let $\omega>1$. Then there exists an absolute constant $c_{0}$
such that the set $S=\{\lambda=a+ib\,:\,|a|\geq c_{0}\omega,\;\;|b|\leq|a|\}$
contains no roots of {\rm(\ref{K1})}.

{\rm(iii)} Let $\omega_{\varepsilon}^{-1}=\omega_{k}^{-1}+\varepsilon$
for a certain $k=2,\,3,\ldots$, then there exist $\gamma_k>0$
and $\varepsilon_{k}>0$ such that for $\varepsilon\in(0,\varepsilon_{k}]$ all solutions
to {\rm(\ref{K1})} which are located in the strip $|\Im\lambda|\leq\gamma_k$
are real described in Sect.~{\rm\ref{sub:sol}} and complex described in Sect.~{\rm\ref{Sect2}}
\footnote{If $\varepsilon\in[-\varepsilon_{k},0]$ then all such solutions are real, see Figure \ref{fig:disp}}.
\end{prop} 

\begin{proof} (i) Assume that $F(\lambda)=\frac{d}{d\lambda}\mathcal{F}(\lambda)=0$.
Consider two cases $\lambda\neq1$ and $\lambda=1$. In the first
case $\kappa\neq0$ and $\sin\kappa=0$, which due to the first equation
in (\ref{K1}) leads to $\cos\kappa=0$ what is impossible. Consider
the second case $\lambda=1$. Then $\mathcal{F}(1)=0$, $\kappa^{2}=\omega^{2}-1$
and $\omega>1$ solves (\ref{K3}).

(ii) Let $\lambda=a+bi\in S$. Then $\kappa=i\lambda(1+O(\omega^{2}\,|\lambda|^{-2}))$.
Since $|\sin\kappa|^{2}=\cosh^{2}\Im\kappa-\cos^{2}\Re\kappa$, we
have 
\[
|\sin\kappa|^{2}-\frac{1}{\omega^{2}}|\kappa|^{2}\geq\cosh^{2}\Im\kappa-1-(1+|a|)^{2}-|a|^{2},
\]
which implies the required assertion.

(iii) Due to (ii) it is sufficient to prove that there are no roots
on the intervals where $\lambda=a\pm i\gamma$ and $|a|\leq C$, where
$C$ is a certain positive constant. We can assume that $\gamma\leq1$.
First, we note that 
$$
\frac{sin\kappa}{\kappa}-\frac{(-1)^{k+1}}{\omega_{\varepsilon}}=g(\lambda,\varepsilon)-\frac{(-1)^{k+1}}{\omega_{\varepsilon}}=\Big(g(a,0)-\frac{(-1)^{k+1}}{\omega_{k}}\Big)\pm i\gamma\frac{\partial g}{\partial\lambda}(a,0)+O(\gamma^{2}+\varepsilon).
$$
If $a\in[-C,1-\delta]\cup[1+\delta,C]$, then $\Big(g(a,0)-\frac{(-1)^{k+1}}{\omega_{k}}\Big)$
and $\frac{\partial g}{\partial\lambda}(a,0)$ are real and 
\[
c(\delta,C)=\max_{[-C,1-\delta]\cup[1+\delta,C]}\Big(|g(a,0)-\frac{(-1)^{k+1}}{\omega_{k}}|+\Big|\frac{\partial g}{\partial\lambda}(a,0)\Big|\Big)|>0.
\]
Furthermore, 
\[
|g(\lambda,\varepsilon)-\frac{(-1)^{k+1}}{\omega_{\varepsilon}}|\geq|\gamma|c(\delta,C)-C_{1}(\gamma^{2}+\varepsilon).
\]
Consider $\lambda$ close to $1$. More exactly, let $|\Im\lambda|=\delta$
and $|a-1|\leq\delta$. Noting that $g(1,0)-\frac{(-1)^{k+1}}{\omega_{k}}=\frac{\partial g}{\partial\lambda}(1,0)=0$
and using representation (\ref{eq:expansion}), we get 
$$
g(\lambda,\varepsilon)-\frac{(-1)^{k+1}}{\omega_{\varepsilon}}=(-1)^{N}\frac{(\lambda-1)^{2}}{2\omega_{k}^{2}(\omega_{k}-1)}+O(\delta^{3}+\varepsilon).
$$
Since $|\lambda-1|^{2}\geq\delta^{2}$, we conclude from the last
estimate that there are positive constants $c_{1}$ and $c_{2}$ depending
only on $k$ such that, if $\delta\leq c_{1}$ and $\varepsilon\leq c_{2}\delta^{2}$,
then $g(\lambda,\varepsilon)-\frac{(-1)^{k+1}}{\omega_{\varepsilon}}$
does not vanish on the intervals $|\Im\lambda|=\delta$, $|a-1|\leq\delta$.
\end{proof}

\subsection{Symplectic form}\label{sub:SF}

There is a natural symplectic structure on the set of solutions to
problem (\ref{homoDirac2}), (\ref{homoDirac3P}), cf. \cite[Ch.\ 5]{NaPl}. It will play important role in the study of the scattering matrix.

For two solutions $w=(u,v)$ and $\tilde{w}=(\tilde{u},\tilde{v})$
of the problem (\ref{homoDirac2}), (\ref{homoDirac3P}), let us define
the quantity
\begin{equation}
q_{a}(w,\tilde{w})=-\int_{0}^{1}(u(x,a)\overline{\tilde{v}(x,a)}-v(x,a)\overline{\tilde{u}(x,a)})dx.\label{eq:qform}
\end{equation}
Since 
\begin{multline*}
0=\int_{\Pi_{a,b}}\left(\begin{array}{c}
\overline{\tilde{u}}\\
\overline{\tilde{v}}
\end{array}\right)\cdot({\mathcal D}-\omega I)\left(\begin{array}{c}
u\\
v
\end{array}\right)dxdy-\int_{\Pi_{a,b}}\left(\begin{array}{c}
u\\
v
\end{array}\right)\cdot(\overline{\mathcal{D}}-\omega I)\left(\begin{array}{c}
\overline{\tilde{u}}\\
\overline{\tilde{v}}
\end{array}\right)dxdy\\
=q_{b}(w,\tilde{w})-q_{a}(w,\tilde{w}),
\end{multline*}
where $\Pi_{a,b}=(0,1)\times(a,b)$,
$a<b$, we see that $q_{a}$ does not depend on $a$ and we will use the notation q for this form.\\
The form q is sesquilinear 
$$
q(\alpha_{1}w_{1},\alpha_{2}w_{2})=\alpha_{1}\overline{\alpha_{2}}q(w_{1},w_{2})
$$
and anti-Hermitian 
$$
\overline{q(w_{1},w_{2})}=-q(w_{2},w_{1})
$$
hence it is symplectic.

\subsection{Biorthogonality conditions}\label{sub:BC}

Here we discuss the biorthogonality conditions for solutions to (\ref{homoDirac2}),
(\ref{homoDirac3P}). Since we are interested mostly in the case when
$\omega=\omega_{\varepsilon}$, where $\omega_{\varepsilon}^{-1}=\omega_{N}^{-1}+\varepsilon$
, we will consider this case here. We introduce solutions to (\ref{homoDirac2}),
(\ref{homoDirac3P}) as follows 
\begin{equation}
w_{j}^{\tau}(x,y)=e^{-i\lambda_{j}^{\tau}y}(\mathcal{U}_{j}^{\tau}(x),\mathcal{V}_{j}^{\tau}(x)),\label{LUD11}
\end{equation}
where $\tau$ stands for $+$ or $-$ and $j=1,\ldots,N$ (if $j=1$,
then only $\tau=-$ is admissible). Furthermore, if $j=1,\ldots,N-1$,
then the functions $\mathcal{U}_{j}^{\tau}$ and $\mathcal{V}_{j}^{\tau}$
are given by (\ref{Uosc2-1}) and in the case $j=N$ they are given
by (\ref{exp}). %The constants $a_k^\tau$ will be chosen later.
Since 
$$
q_{a}(w_{j}^{\tau},w_{k}^{\theta})=e^{-i(\lambda_{j}^{\tau}-\overline{\lambda_{k}^{\theta}})a}C_{jk}^{\tau\theta},
$$
where $C_{jk}^{\tau\theta}$ is a constant, and since the form $q$
is independent of $a$, we conclude that 
$$
q(w_{j}^{\tau},w_{k}^{\theta})=0\;\;\mbox{if \ensuremath{(j,\tau)\neq(k,\theta)}\;\;\mbox{and \ensuremath{q(w_{j}^{\tau},w_{j}^{\tau})=\frac{i}{\omega}(\lambda_{j}^{\tau}-1)}}}.$$
Therefore, 
\begin{equation}
q(w_{j}^{\tau},w_{k}^{\theta})=\frac{i}{\omega}(\lambda_{k}^{\tau}-1)\delta_{j,k}\delta_{\tau,\theta}\label{April1a}
\end{equation}
for $j,k=1,\ldots,N-1$ and $\tau,\theta=\pm$.

Let us start with the oscillatory waves.
We put 
\begin{equation*}
{\bf w}_{k}^{\tau}=\frac{\sqrt{\omega}}{\sqrt{|\lambda_{k}^{\tau}-1|}}w_{k}^{\tau},\;\;\; k=1,\ldots,N-1.
\end{equation*}
Then by (\ref{April1a}) 
\begin{equation}
q({\bf w}_{j}^{\tau},{\bf w}_{k}^{\theta})=\tau i\delta_{j,k}\delta_{\tau,\theta},\,\tau,\theta=\pm.\label{eq:bioosc}
\end{equation}
In the case $\omega=\omega_{N}$, we have 
$$
q(w_{N}^{0},w_{N}^{0})=0,\ q(w_{N}^{0},w_{N}^{1})=\frac{1}{2\omega_{N}},\; q(w_{N}^{1},w_{N}^{1})=\frac{i}{6\omega_{N}(\omega_{N}^{2}-1)},
$$
for waves (\ref{eq:sol01-1}), (\ref{eq:sol02}).

\subsection{Biorthogonality conditions for the complex wave vector~$\lambda$}\label{sub:BCl}

Let us check if the waves $\mathbf{w}_{N}^{\varepsilon\pm}$ fullfil
orthogonality conditions. For waves $w_{N}^{\pm}$ we have
\begin{equation}
q(w_{N}^{\pm},w_{N}^{\pm})=0,\ \ \ q(w_{N}^{+},w_{N}^{-})=\frac{i}{\omega}(\lambda_{+}-1),\ \ \ q(w_{N}^{-},w_{N}^{+})=\frac{i}{\omega}(\lambda_{-}-1).\label{eq:qeps}
\end{equation}
We put 
\begin{equation}
a^{\varepsilon}:=iq(\mathbf{w}_{N}^{\varepsilon+},\mathbf{w}_{N}^{\varepsilon+}),\quad b^{\varepsilon}:=-iq(\mathbf{w}_{N}^{\varepsilon+},\mathbf{w}_{N}^{\varepsilon-}),\quad c^{\varepsilon}:=-iq(\mathbf{w}_{N}^{\varepsilon-},\mathbf{w}_{N}^{\varepsilon-})\label{April1b}
\end{equation}
Then using (\ref{eq:qeps}) together with (\ref{eq:we+}), (\ref{eq:we-})
and (\ref{eq:gammaapprox}) we obtian 
\begin{equation}
a^{\varepsilon}\rightarrow a^{0}:=2\omega_{N}(\omega_{N}^{2}-1)\frac{5+9\omega_{N}^{2}}{3},\label{April1c}
\end{equation}

\begin{equation}
b^{\varepsilon}\rightarrow b^{0}:=2\omega_{N}(\omega_{N}^{2}-1)^{2},\label{April1d}
\end{equation}

\begin{align*}
c^{\varepsilon} & =\varepsilon4\omega_{N}(\omega_{N}^{2}-1)^{2}\frac{9\omega_{N}^{2}-3\omega_{N}+7}{3}+O(\varepsilon^{2})\:\:\mbox{as}\:\:\varepsilon\rightarrow0\label{April1e}
\end{align*}
The functions $a^{\varepsilon}$, $b^{\varepsilon}$ and $c^{\varepsilon}$ are
real and analytic, and $a^{\varepsilon},b^{\varepsilon},c^{\varepsilon}>0$
for small $\varepsilon>0$.

From the above evaluations we see that waves $\mathbf{w}_{N}^{\varepsilon+}$ and
$\mathbf{w}_{N}^{\varepsilon-}$ do not fullfil the biorthogonality conditions.
That is why we consider their linear combinations 
\begin{equation}
{\bf w}_{N}^{+}=\frac{\mathbf{w}_{N}^{\varepsilon+}+\alpha_{\varepsilon}\mathbf{w}_{N}^{\varepsilon-}}{N_{1}^{\varepsilon}},
\;\;\;
{\bf w}_{N}^{-}=\frac{\mathbf{w}_{N}^{\varepsilon+}}{N_{2}^{\varepsilon}}
\label{mf2}
\end{equation}
where $\alpha_{\varepsilon}$ is unknown constant and $N_{1}^{\varepsilon}$
and $N_{2}^{\varepsilon}$ are normalizing factors.

Our aim is to fulfil the biorthogonality relations 
\begin{equation}
q({\bf w}_{N}^{\tau},{\bf w}_{N}^{\theta})=\tau i\delta_{\tau,\theta},\,\tau,\theta=\pm,\label{eq:qnew}
\end{equation}
which implies, in particular, that 
\[
q(\mathbf{w}_{N}^{\varepsilon+},\mathbf{w}_{N}^{\varepsilon+})+\alpha_{\varepsilon}q(\mathbf{w}_{N}^{\varepsilon-},\mathbf{w}_{N}^{\varepsilon+})=0.
\]
Therefore, using (\ref{April1b}), (\ref{April1c}) and (\ref{April1d}), we get  
\begin{equation*}
\alpha_{\varepsilon}=\frac{a^{\varepsilon}}{b^{\varepsilon}}=\frac{a^{0}}{b^{0}}+O(\varepsilon).
\end{equation*}
From (\ref{eq:qnew}) and (\ref{April1b}), we find also the normalization factors $N_{1}^{\varepsilon}$
and $N_{2}^{\varepsilon}$
\begin{equation*}
N_{1}^{\varepsilon}=\sqrt{a^{0}}+O(\varepsilon)\quad N_{2}^{\varepsilon}=\sqrt{a^{\varepsilon}}=\sqrt{a^{0}}+O(\varepsilon)
\end{equation*}
By (\ref{eq:we+}), (\ref{eq:we-}) and (\ref{mf2}),
we have 
\begin{equation}
{\bf w}_{N}^{+}=\alpha_{1}w_{N}^{+}+\beta_{1}w_{N}^{-}\label{mf2-2}
\end{equation}
and 
\begin{equation}
{\bf w}_{N}^{-}=\alpha_{2}w_{N}^{+}+\beta_{2}w_{N}^{-},\label{mf1-1-2}
\end{equation}
where

\begin{equation}
\alpha_{1}=\frac{\gamma_{+}^{\varepsilon}(1+\alpha_{\varepsilon}(\lambda_{+}-1))}{N_{1}^{\varepsilon}},\,\beta_{1}=\frac{\gamma_{-}^{\varepsilon}(1+\alpha_{\varepsilon}(\lambda_{-}-1))}{N_{1}^{\varepsilon}}\label{eq:alpha1}
\end{equation}
and 
\begin{equation}
\alpha_{2}=\frac{\gamma_{+}^{\varepsilon}}{N_{2}^{\varepsilon}},\,\beta_{2}=\frac{\gamma_{-}^{\varepsilon}}{N_{2}^{\varepsilon}}.\label{eq:alpha2}
\end{equation}
Then biorthogonality conditions take the form (\ref{eq:qnew}).

\subsection{Properties of coefficients (\ref{eq:alpha1}) and (\ref{eq:alpha2}) }\label{sub:PC}

Accoriding to definitions (\ref{eq:alpha1}) and (\ref{eq:alpha2}),
one can check that 
\begin{equation}
\overline{\beta_{1}}=\alpha_{1},\;\;\;\overline{\beta_{2}}=\alpha_{2}.\label{TT2}
\end{equation}
In the next proposition, we collect some more properties of coefficients
(\ref{eq:alpha1}) and (\ref{eq:alpha2}), which will play an important
role in the sequel.

\begin{prop}\label{Pto2s} The following relations hold 
\begin{equation}
\frac{\alpha_{1}}{\alpha_{2}}=\frac{\beta_{2}}{\beta_{1}},\;\;\;\;\;\;\;\Big|\frac{\alpha_{1}}{\alpha_{2}}\Big|=1.\label{TT3}
\end{equation}
\end{prop} \begin{proof} From (\ref{eq:qnew}) it follows

\begin{equation}
\alpha_{1}\overline{\beta_{1}(}\lambda_{+}-1)+\beta_{1}\overline{\alpha_{1}(}\lambda_{-}-1)=\omega,\;\;\;\;\;\\
\alpha_{2}\overline{\beta_{2}(}\lambda_{+}-1)+\beta_{2}\overline{\alpha_{2}(}\lambda_{-}-1) =-\omega,\label{eq:ab_ort}
\end{equation}
and 
\begin{equation}
\alpha_{1}\overline{\beta_{2}(}\lambda_{+}-1)+\beta_{1}\overline{\alpha_{2}(}\lambda_{-}-1)=0,\label{TT4}
\end{equation}
Eliminating $\omega$ from (\ref{eq:ab_ort}), we get

\begin{equation*}
\frac{\lambda_{-}-1}{\lambda_{+}-1}=-\frac{\alpha_{1}\overline{\beta_{1}}+\alpha_{2}\overline{\beta_{2}}}{\beta_{1}\overline{\alpha_{1}}+\beta_{2}\overline{\alpha_{2}}}=-\frac{\alpha_{1}^{2}+\alpha_{2}^{2}}{\overline{\alpha_{1}^{2}}+\overline{\alpha_{2}^{2}}}
\end{equation*}
where the last equality follows form (\ref{TT2}). Inserting it into (\ref{TT4}), we arrive at
\begin{equation}
\alpha_{1}\alpha_{2}-\overline{\alpha_{1}}\overline{\alpha_{2}}\frac{\alpha_{1}^{2}+\alpha_{2}^{2}}{\overline{\alpha_{1}^{2}}+\overline{\alpha_{2}^{2}}}=0.\label{eq:a1a2_relation}
\end{equation}
Dividing (\ref{eq:a1a2_relation}) by $\alpha_{2}(\overline{\alpha_{2}})^{2}$ and multiplying
by $\overline{\alpha_{1}^{2}}+\overline{\alpha_{2}^{2}}$, we obtain
\[
\frac{\alpha_{1}}{\alpha_{2}}+\frac{\alpha_{1}}{\alpha_{2}}(\frac{\overline{\alpha_{1}}}{\overline{\alpha_{2}}})^{2}-\frac{\overline{\alpha_{1}}}{\overline{\alpha_{2}}}(\frac{\alpha_{1}}{\alpha_{2}})^{2}-\frac{\overline{\alpha_{1}}}{\overline{\alpha_{2}}}=0.
\]
Defining $d=\alpha_{1}/\alpha_{2}$ , this relation can be
written as
\[
(d-\overline{d})(1-d\overline{d})=0.
\]
Here $d\neq\overline{d}$ because otherwise $\lambda_{+}$ would be real,
this follows from the definitions of $\alpha_{1}$ and $\alpha_{2}$.
Consequently
\begin{equation}
(1-d\overline{d})=0,\label{eq:TT5}\,\,\,|d|=1,
\end{equation}
which implies the second equality in (\ref{TT3}).
To prove first equality in (\ref{TT3}) we note that $d=\alpha_{1}/\alpha_{2}=\overline{\beta_{1}}/\overline{\beta_{2}}$
by (\ref{TT2}). This together with (\ref{eq:TT5}) gives
\[
1-\frac{\alpha_{1}}{\alpha_{2}}\frac{\beta_{1}}{\beta_{2}}=0,\,\,\,
\frac{\alpha_{1}}{\alpha_{2}}=\frac{\beta_{2}}{\beta_{1}}.
\]
The proof is completed. \end{proof}

The following quantity will play an important role 
\begin{equation}
d=d(\varepsilon):=\frac{\alpha_{1}}{\alpha_{2}}=\frac{\beta_{2}}{\beta_{1}},\label{Apr9aa}
\end{equation}
where the last equality is borrowed from Proposition \ref{Pto2s}. By (\ref{eq:TT5}), the absolute value of $d$ is equal to $1$.
Moreover the function $d(\varepsilon)$ depends on $\varepsilon\in[0,\varepsilon_{0}]$
and by definitions of \newline$\alpha_{1}$ and $\alpha_{2}$, see (\ref{eq:alpha1})
and (\ref{eq:alpha2}) 
\begin{equation}
d(\varepsilon)=1+i\frac{(5+9\omega_{N}^{2})}{3}\sqrt{\frac{2\omega_{N}}{\omega_{N}^{2}-1}}\sqrt{\varepsilon}+O(\varepsilon).\label{eq:dexpansion}
\end{equation}
We note also that 
\begin{equation}
\mathbf{w}_{N}^{+}=Aw_{N}^{0}+iBw_{N}^{1}+O(\varepsilon)\label{eq:w+approx}
\end{equation}
and 
\begin{equation}
\mathbf{w}_{N}^{-}=Cw_{N}^{0}+iBw_{N}^{1}+O(\varepsilon)\label{eq:w-approx}
\end{equation}
where
\[
A=(-1)^{N}(-6\omega_{N}^{5}+2\omega_{N}^3+9\omega_N^2+4\omega_N+5)
\sqrt{\frac{2\omega_N(\omega_{N}^{2}-1)}{3(9\omega_{N}^{2}+5)}},
\]
\[
B=(-1)^{N+1}\sqrt{\frac{6\omega_N(\omega_{N}^{2}-1)}{9\omega_{N}^{2}+5}},\:\:\: C=(-1)^{N+1}2\omega_{N}(3\omega_{N}^{2}+2)\sqrt{\frac{2\omega_N(\omega_{N}^{2}-1)}{3(9\omega_{N}^{2}+5)}}.
\]

\subsection{Non-homogeneous problem}\label{sub:NH}

Here, we consider the non-homogeneous problem
\begin{eqnarray}
(-i\partial_{x}+\partial_{y})v-\omega u=g\;\;\;\mbox{in \ensuremath{\Pi},}\label{nonH1}\\
(-i\partial_{x}-\partial_{y})u-\omega v=h\;\;\;\mbox{in\,\ \ensuremath{\Pi},}
\label{nonH2}
\end{eqnarray}
supplied with the boundary conditions 
\begin{equation}
u(0,y)=0\;\;\mbox{and}\;\, v(1,y)=0,\;\; y\in\Bbb R.\label{2a}
\end{equation}
To study solvability of this system, we introduce some spaces. The
space $L_{\sigma}^{\pm}(\Pi)$, $\sigma>0$, consists of all functions
$g$ such that $e^{\pm\sigma y}g\in L^{2}(\Pi)$. Furthermore, 
\[
X_{\sigma}^{\pm}=\{u\in L_{\sigma}^{\pm}(\Pi)\,:\,(i\partial_{x}+\partial_{y})u\in L_{\sigma}^{\pm}(\Pi),\;\; u(0,y)=0\},
\]
\[
Y_{\sigma}^{\pm}=\{v\in L_{\sigma}^{\pm}(\Pi)\,:\,(-i\partial_{x}+\partial_{y})v\in L_{\sigma}^{\pm}(\Pi),\;\; v(1,y)=0\}.
\]
The norms in the above spaces are defined by $||g;L_{\sigma}^{\pm}(\Pi)||=||e^{\pm\sigma y}g;L^{2}(\Pi)||$,
\[
||u;X_{\sigma}^{\pm}||=\Big(||u;L_{\sigma}^{\pm}(\Pi)||^{2}+||(i\partial_{x}+\partial_{y})u;L_{\sigma}^{\pm}(\Pi)||^{2}\Big)^{1/2}
\]
and 
\[
||v;Y_{\sigma}^{\pm}||=\Big(||v;L_{\sigma}^{\pm}(\Pi)||^{2}+||(-i\partial_{x}+\partial_{y})v;L_{\sigma}^{\pm}(\Pi)||^{2}\Big)^{1/2}
\]
respectively.

The main solvability result is the following assertion
\begin{theorem}\label{T1}
Let $\omega>0$ and let $\sigma>0$ be such that the line $\Im\lambda=\pm\sigma$
contains no roots of {\rm(\ref{K1})}. Then the operator 
\footnote{For the simplicity of the notation, we will be writing $L_{\sigma}^{\pm}(\Pi)$
for both spaces of functions and vectors. Here for example we write
$L_{\sigma}^{\pm}(\Pi)$ instead of $L_{\sigma}^{\pm}(\Pi)\times L_{\sigma}^{\pm}(\Pi)$.
This notation will be applied to the other spaces introduced later
as well. %
} 
\begin{equation}
D-\omega I\;:\; X_{\sigma}^{\pm}\times Y_{\sigma}^{\pm}\;\to\; L_{\sigma}^{\pm}(\Pi)\label{2aa}
\end{equation}
is isomorphism. 
\end{theorem} 

Despite the problem, which we are dealing with is not elliptic the proof of this assertion is more or less standard and we present it in Appendix \ref{sec:Appendix:Proof1}.

In what follows we assume that an integer $N\geq2$ is fixed, $\omega_{N}=\sqrt{\kappa_{N}^{2}+1}$,
where $\kappa_{N}$ is defined in Sect.~\ref{sub:sol}, and $\omega=\omega_{\varepsilon}=\omega_{N}/(1+\varepsilon\omega_{N})$.
Then we have the following waves 
\[
{\bf w}_{1}^{-},{\bf w}_{2}^{\pm},\ldots,{\bf w}_{N-1}^{\pm},{\bf w}_{N}^{\pm},
\]
where ${\bf w}_{1}^{-}$ corresponding to $\lambda_{1}^{-}$ and ${\bf w}_{j}^{\pm}$,
$j=2,\dots,N-1$ are oscillatory and ${\bf w}_{N}^{\pm}$ are of exponential
growth.

\begin{theorem}\label{T1s} Let $\gamma=\gamma_N$ and $\varepsilon_{N}$ be
the same positive numbers as in Proposition {\rm\ref{Pr1}} and also $(g,h)\in L_{\gamma}^{+}(\Pi)\cap L_{\gamma}^{-}(\Pi)$.
Denote by $(u^{\pm},v^{\pm})\in X_{\gamma}^{\pm}\times Y_{\gamma}^{\pm}$
the solution of problem {\rm(\ref{nonH1})}, {\rm(\ref{nonH2})}, {\rm(\ref{2a})}, which
exist according to {\rm Theorem} {\rm\ref{T1}}. Then 
\begin{equation}
(u^{+},v^{+})=(u^{-},v^{-})+\sum_{j=2}^{N}C_{j}^{+}{\bf w}_{j}^{+}+\sum_{j=1}^{N}C_{j}^{-}{\bf w}_{j}^{-},\label{Apr8a}
\end{equation}
where 
\begin{equation*}
-iC_{j}^{+}=\int_{\Pi}\,(g,h)\cdot\overline{{\bf w}_{j}^{+}}dxdy,\;\; iC_{j}^{-}=\int_{\Pi}\,(g,h)\cdot\overline{{\bf w}_{j}^{-}}dxdy.\label{Apr8b}
\end{equation*}
\end{theorem} 
Proof of this Theorem is presented in Appendix \ref{sec:Appendix:Proof2}. Results similar to Theorems \ref{T1} and \ref{T1s} are well-known for elliptic problems and can be found, for example, in \cite{KozlovBook}, \cite{KozlovBookElliptic} and \cite{NaPl}, but we remind that our problem is not elliptic, see Appendix \ref{sec:Appendix:-Ellipticity}.

\section{Dirac equation with potential\label{sub:Perturbated-Dirac-equation}}

\subsection{Problem statement}\label{PS2}

Here we examine the problem with a potential, prove solvability results
and asymptotics formulas for solutions. \\
Consider the nanoribbon with a potential: 
\begin{eqnarray}
{\mathcal D}\left(\begin{array}{c}
u\\
v
\end{array}\right)+\delta\mathcal{P}\left(\begin{array}{c}
u\\
v
\end{array}\right)=\omega\left(\begin{array}{c}
u\\
v
\end{array}\right),\label{DiracP1}\\
u(0,y)=0\ ,\ \ v(1,y)=0,\label{Dirac}
\end{eqnarray}
where ${\mathcal D}$ is the same as in (\ref{homoDirac2}), $\mathcal{P}=\mathcal{P}(x,y)$
is a bounded, continuous and real-valued function with compact
support in $\overline{\Pi}$ and $\delta$ is a small parameter. We
assume in what follows that 
\begin{equation}
\mbox{supp}\mathcal{P}\subset[-R_{0},R_{0}]\times[0,1]\:\:\:\mbox{and}\:\:\:\sup_{(x,y)\in\Pi}|\mathcal{P}(x,y)|\leq1,\label{eq:supP}
\end{equation}
where $R_{0}$ is a fixed positive number.

We assume that positive numbers $\gamma$ and $\varepsilon_{N}$ are
fixed such that 

\begin{equation}
\omega=\omega_{\varepsilon}\;\;\;\mbox{where}\;\;\;\frac{1}{\omega_{\varepsilon}}=\frac{1}{\omega_{N}}+\varepsilon,\; N=2,\,3,\ldots,\;\;\label{DD1}
\end{equation}
$\varepsilon\in[0,\varepsilon_{N}]$ and $\gamma=\gamma_N$ is
from Proposition \ref{Pr1}(iii), i.e. 

(1) the lines $\Im\lambda=\pm\gamma$ contain no roots of (\ref{K1})
with $\omega$ given by (\ref{DD1}), 

(2) the strip $|\Im\lambda|<\gamma$ contains roots of (\ref{K1}),
which are real and complex described in Sect. \ref{sub:sol} and \ref{Sect2} respectively. 

Since the norm of the multiplication operator $\delta\mathcal{P}$
in $L^{2}(\Pi)$ is less than $\delta$ we derive from Theorem \ref{T1}
the following assertion

\begin{theorem}\label{T1a} The operator 
\begin{equation}
{\mathcal D}+(\delta\mathcal{P}-\omega_{\varepsilon}) I\;:\; X_{\gamma}^{\pm}\times Y_{\gamma}^{\pm}\;\to\; L_{\gamma}^{\pm}(\Pi)\label{2aas}
\end{equation}
is isomorphism for $|\delta|\leq\delta_{0}$, where $\delta_{0}$
is a positive constant depending on the norm on the inverse operator
$({\mathcal D}-\omega_{\varepsilon}I)^{-1}\,:\, L_{\gamma}^{\pm}(\Pi)\to L_{\gamma}^{\pm}(\Pi)$.
\end{theorem}

We introduce two new spaces 
\[
{\mathcal H}_{\gamma}^{+}=\{(u,v)\;:\; u\in X_{\gamma}^{+}\cap X_{\gamma}^{-},\;\; v\in Y_{\gamma}^{+}\cap Y_{\gamma}^{-}\}
\]
and 
\[
{\mathcal H}_{\gamma}^{-}=\{(u,v)\;:\; u\in X_{\gamma}^{+}\cup X_{\gamma}^{-},\;\; v\in Y_{\gamma}^{+}\cup Y_{\gamma}^{-}\}.
\]
The norms in this spaces are defined by 
\[
||(u,v);{\mathcal H}_{\gamma}^{\pm}||^{2}=\int_{\Pi}e^{\pm2\gamma|y|}\Big(|u|^{2}+|v|^{2}+|{\mathcal D}(u,v)^{t}|^{2}\Big)dxdy.
\]
Let also ${\mathcal L}_{\gamma}^{\pm},$ $\gamma>0$, be two $L^{2}$
weighted spaces in $\Pi$ with the norms 
\[
||(u,v);{\mathcal L}_{\gamma}^{\pm}||^{2}=\int_{\Pi}e^{\pm2\gamma|y|}\Big(|u|^{2}+|v|^{2}\Big)dxdy.
\]
We define two operators acting in the introduced spaces 
\begin{equation*}
A_{\gamma}^{\pm}=A_{\gamma}^{\pm}(\varepsilon,\delta)={\mathcal D}+(\delta\mathcal{P}-\omega) I\,:\,{\mathcal H}_{\gamma}^{\pm}\rightarrow{\mathcal L}_{\gamma}^{\pm}.
\end{equation*}
Some important properties of these operator are collected in the following

\begin{theorem}\label{T3s} The operators $A_{\gamma}^{\pm}$ are
Fredholm and ${\rm dim}\,{\rm ker}A_{\gamma}^{+}=0$, ${\rm dim}\,{\rm coker}\,A_{\gamma}^{-}=0$.
Moreover 
\[
{\rm dim}\,{\rm coker}A_{\gamma}^{+}={\rm dim}\,{\rm ker}A_{\gamma}^{-}=2N-1,
\]
where $N$ is the same as in {\rm(\ref{DD1})}. \end{theorem}

\begin{proof} By Theorem \ref{T1a}, the operator (\ref{2aas}) is
isomorphic. This implies that the operator $A_{\gamma}^{-}$ is surjective
for such $\delta$ and its index and the dimension of his kernel does
not depend on $\delta$ and hence equals $2N-1$ as it is in the case
$\delta=0$.\end{proof}

In the next theorem and in what follows, we fix
four smooth functions, $\chi_{\pm}=\chi_{\pm}(y)$ and $\eta_{\pm}=\eta_{\pm}(y)$
such that $\chi_{+}(y)=1$, $\chi_{-}(y)=0$ for $y>R_{0}$ and $\chi_{+}(y)=0$,
$\chi_{-}(y)=1$ for $y<-R_{0}$. Then let $\eta_{\pm}(y)=1$ for
large positive $\pm y$, $\eta_{\pm}(y)=0$ for large negative $\pm y$
and $\chi_{\pm}\eta_{\pm}=\chi_{\pm}$.

Let us derive an asymptotic formula for the solution to the perturbed
problem (\ref{DiracP1}), (\ref{Dirac}).

\begin{theorem}\label{T3s2} Let $f\in{\mathcal L}_{\gamma}^{+}$ and
let $w=(u,v)\in{\mathcal H}_{\gamma}^{-}$ be a solution to 
\begin{equation}
({\mathcal D}+(\delta\mathcal{P}-\omega) I)w=f.\label{April9b}
\end{equation}
satisfying {\rm(\ref{Dirac})}.Then 
\begin{equation}
w=\eta_{+}\sum_{j=1}^{N}\sum_{\tau=\pm}C_{j}^{\tau}{\bf w}_{j}^{\tau}+\eta_{-}\sum_{j=1}^{N}\sum_{\tau=\pm}D_{j}^{\tau}{\bf w}_{j}^{\tau}+R,\label{April9c}
\end{equation}
where $R\in{\mathcal H}_{\gamma}^{+}$ . \end{theorem}

\begin{proof}We write (\ref{April9b}) as
\[
({\mathcal D}+\omega I)w=f-\delta\mathcal{P}w=:F
\]
Then
\begin{equation}
({\mathcal D}+\omega I)\eta_{\pm}w=\eta_{\pm}F+[{\mathcal D},\eta_{\pm}]w.\label{eq:Fcomm}
\end{equation}
According to Theorem \ref{T1s} solutions to (\ref{eq:Fcomm})
are $w^{\mp}:=(u^{\mp},v^{\mp})\in X_{\gamma}^{\mp}\times Y_{\gamma}^{\mp}$,
so
\begin{equation}
\eta_{+}w=w^{-},\:\:\:\eta_{-}w=w^{+}.\label{eq:etasol}
\end{equation}
Applying relation (\ref{Apr8a}) to (\ref{eq:etasol}), then
multiplying obtained equations by $\chi_{\pm}$ respectively and summing
them up, we get
\begin{equation}
w=\chi_{-}\sum_{j=1}^{N}\sum_{\tau=\pm}C_{j}^{'\tau}{\bf w}_{j}^{\tau}+\chi_{+}\sum_{j=1}^{N}\sum_{\tau=\pm}D_{j}^{'\tau}{\bf w}_{j}^{\tau}+R'\label{eq:wchi}
\end{equation}
where $R'=\chi_{+}w^{+}+\chi_{-}w^{-}+(1-\chi_{+}-\chi_{-})w\in\mathcal{H_{\gamma}}^{+}$.
Equation (\ref{eq:wchi}) can be written in the form (\ref{April9c})
with $R\in\mathcal{H_{\gamma}}^{+}$.\end{proof}

\subsection{Augumented scattering matrix}\label{sub:S}

Scattering matrix is our main tool for the identification of trapped
modes \cite{TMP, na510, FAA}. Using q-form, we define incoming/outgoing waves. Scattering
matrix is defined via coefficients in this combination of waves. It
is important to point out that this matrix is often called augumented
as it contains coefficients of the waves exponentially growing at
infinity as well. Finally, by the end of the section we define a space
with separated asymptotics and check that it produces a unique solution
to the perturbed problem.

Let
\begin{equation*}
Q_{R}(w,\tilde{w})=q_{R}(w,\tilde{w})-q_{-R}(w,\tilde{w}).\label{W1}
\end{equation*}
If $w=(u,v)$ and $\tilde{w}=(\tilde{u},\tilde{v})$ are solutions
to (\ref{homoDirac2}) for $|y|\geq R_{0}$ then using Green's formula
one can show that this form is independent of $R$, $R\geq R_{0}$.

We introduce two sets of waves with cutoff close to $\pm\infty$, which we
will call outgoing and incoming (for physical interpretation see Appendix
\ref{sec:Appendix:-Mandelstam-radiation}) 
\begin{equation}
W_{k}^{\mp}=W_{k}^{\mp}(x,y;\varepsilon)=\chi_{\pm}(y){\bf w}_{k}^{\mp}(x,y)\label{W1a}
\end{equation}
and 
\begin{equation}
V_{k}^{\pm}=V_{k}^{\pm}(x,y;\varepsilon)=\chi_{\pm}(y){\bf w}_{k}^{\pm}(x,y)\label{W1b}
\end{equation}
with $k=2,\ldots,N$ for the sign $+$ and $k=1,\ldots,N$ for the sign
$-$. The reason for introducing this sets of waves is the following
property 
\begin{equation}
Q_{R}(W_{k}^{\tau},W_{j}^{\theta})=-i\delta_{k,j}\delta_{\tau,\theta},\;\; Q_{R}(V_{k}^{\tau},V_{j}^{\theta})=i\delta_{k,j}\delta_{\tau,\theta}\label{W1c}
\end{equation}
where $k,j=2,\ldots,N$, $\tau,\theta=\pm$ and $k,j=1$, $\tau,\theta=-$.
Moreover, 
\begin{equation}
Q_{R}(W_{k}^{\tau},V_{j}^{\theta})=0.\label{W1ca}
\end{equation}
Thus the sign of the $Q$-form distinguish among $W$ and $V$ waves.

In the next lemma we give a description of the kernel of the operator
$A_{\gamma}^{-}$, which will be used in the definition of the scattering
matrix.

\begin{theorem}\label{tTh1} There exists a basis in $\ker A_{\gamma}^{-}$
of the form 
\begin{equation}
z_{k}^{\tau}=V_{k}^{\tau}+{\bf S}_{k\tau}^{1-}W_{1}^{-}+\sum_{\theta=\pm}\sum_{j=2}^{N}{\bf S}_{k\tau}^{j\theta}W_{j}^{\theta}+\tilde{z}_{k}^{\tau},\label{zk}
\end{equation}
where $\tilde{z}_{k}^{\tau}\in{\mathcal H}_{\gamma}^{+}$. Moreover, the
coefficients ${\bf S}_{k\tau}^{j\theta}={\bf S}_{k\tau}^{j\theta}(\varepsilon,\delta)$
are uniquely defined %
\footnote{As before we assume that for $k=1$ the only admissible sign is $\tau=-$
and similar agreement is valid for $j=1$, for simplicity, in what
follows we often write summations over all indices $\tau=\pm$ and
$j=1,\,2,..,N$ even though the sign $\tau=+$ should be omitted for $j=1$. %
}. \end{theorem} 

\begin{proof} Let $z\in\ker A_{\gamma}^{-}$. Since 
\[
({\mathcal D}-\omega)(\chi_{\pm}z)=-\delta\mathcal{P}\chi_{\pm}z+[{\mathcal D},\chi_{\pm}]z,
\]
applying Theorem \ref{T1s}, we get 
\[
\chi_{\pm}z=\sum_{j=1}^{N}\sum_{\tau=\pm}a_{j\tau}^{\pm}{\bf w}_{j}^{\tau}+R_{\pm},\;\; R_{\pm}\in X_{\gamma}^{\pm}\times Y_{\gamma}^{\pm}.
\]
Multiplying these equalities by $\chi_{\pm}$ and summing them up,
we get 
\[
z=\sum_{j=1}^{N}\sum_{\tau=\pm}a_{j\tau}^{\pm}\chi_{\pm}{\bf w}_{j}^{\tau}+R,\;\; R=\chi_{+}R_{+}+\chi_{-}R_{-}+(1-\chi_{-}-\chi_{+})z\in{\mathcal H}_{\gamma}^{+}.
\]
We write the last relation in the form 
\[
z=\sum_{j=1}^{N}\sum_{\tau=\pm}C_{j\tau}^{1}W_{j}^{\tau}+\sum_{j=1}^{N}\sum_{\tau=\pm}C_{j\tau}^{2}V_{j}^{\tau}+R,\;\; R\in{\mathcal H}_{\gamma}^{+}
\]
and consider the map 
\[
\ker A_{\gamma}^{-}\ni z\mapsto\,\{C_{j\tau}^{2}\}\in\Bbb C^{2N-1},
\]
which is linear and is denoted by ${\mathcal J}$. Let us show that it's
kernel is trivial. Indeed, if all $C_{j\tau}^{2}$ vanish, then 
\[
\int_{_{\Pi_{R}}}\overline{z}({\mathcal D}+(\delta\mathcal{P}-\omega) I)zdxdy=Q_{R}(z,z)\mapsto-i\sum|C_{j\tau}^{1}|^{2}\;\;\mbox{as \ensuremath{R\to\infty},}
\]
where $\Pi_{R}=(-R,R)\times(0,1)$. Hence $C_{j\tau}^{1}=0$, which
leads to $z\in{\mathcal H}_{\gamma}^{+}$ and therefore $z=0$. This shows
that the mapping ${\mathcal J}$ is invertible and we obtain existence
of a basis in the form (\ref{zk}) together with uniqueness of coefficients ${\bf S}_{k\tau}^{j\theta}$.
\end{proof}

The matrix of coefficients ${\bf S}_{k\tau}^{j\theta}$ in (\ref{zk}) is called the scattering matrix (see the footnote on the previous page concerning $k=1$ and $j=1$).

\subsection{Block notation}\label{sub:BN}

We shall use a vector notation 
\[
{\bf W}=({\bf W}_{\bullet},{\bf W}_{\dagger}),{\bf \;\;\; V}=({\bf V}_{\bullet},{\bf V}_{\dagger}),
\]
where 
\[
{\bf W}_{\bullet}=(W_{1}^{-},W_{2}^{+},W_{2}^{-},\ldots,W_{N-1}^{+},W_{N-1}^{-}),\;\;\;{\bf W}_{\dagger}=(W_{N}^{+},W_{N}^{-})
\]
and
\[
{\bf V}_{\bullet}=(V_{1}^{-},V_{2}^{+},V_{2}^{-},\ldots,V_{N-1}^{+},V_{N-1}^{-}),\;\;\;{\bf V}_{\dagger}=(V_{N}^{+},V_{N}^{-}).
\]
Equation (\ref{zk}) in the vector form reads\footnote{In this format {\bf z}, {\bf V}, {\bf W} and {\bf r} are column vectors. Such notation will follow through the paper.}\label{z_sol} 
\begin{equation}
{\bf z}={\bf V}+{\bf S}{\bf W}+{\bf r}.
\end{equation}
with 
\begin{equation*}
{\bf z}=({\bf z}_{\bullet},{\bf z}_{\dagger})=(z_{1}^{-},z_{2}^{+},z_{2}^{-},..,z_{N-1}^{+},z_{N-1}^{-},z_{N}^{+},z_{N}^{-})\in{\mathcal H}_{\gamma}^{-},\label{z_solb}
\end{equation*}
\begin{equation*}
{\bf r}=({\bf r}_{\bullet},{\bf r}_{\dagger})=(r_{1}^{-},r_{2}^{+},r_{2}^{-},..,r_{N-1}^{+},r_{N-1}^{-},r_{N}^{+},r_{N}^{-})\in{\mathcal H}_{\gamma}^{+}
\end{equation*}
here both vectors ${\bf z}$ and $r$ has $2N-1$ elements. Matrix
${\bf S}={\bf S}(\varepsilon,\delta)$ is written blockwise
\begin{equation*}
\mathcal{S}=\left(\begin{array}{cc}
{\bf S}_{\bullet\bullet} & {\bf S}_{\bullet\dagger}\\
{\bf S}_{\dagger\bullet} & {\bf S}_{\dagger\dagger}
\end{array}\right).\label{Apr9a}
\end{equation*}
Relations (\ref{W1c}) and (\ref{W1ca}) take the form 
\begin{equation}
Q({\bf W},{\bf W})=-i\mathbb{I}\ ,\ \ Q({\bf V},{\bf V})=i\mathbb{I}\ \ \ Q({\bf W},{\bf V})=\mathbb{O}.\label{eq:Qprop}
\end{equation}
where $\mathbb{I}$ is the identity matrix and $\mathbb{O}$ is the
null matrix. 

\begin{prop}\label{sdeltaunitary} The scattering matrix ${\bf S}$
is unitary. \end{prop}

\begin{proof} Since $z_{k}^{\tau}$ satisfies the homogeneous equation
(\ref{Dirac}), by using the Green formula one can show that $Q({\bf z},{\bf z})=0$.
Therefore
\[
0=Q({\bf z},{\bf z})=Q({\bf V}+{\bf S}{\bf W},{\bf V}+{\bf S}{\bf W})=Q({\bf V},{\bf V})+Q({\bf S}{\bf W},{\bf S}{\bf W})=i\mathbb{I}-i{\bf S}^{*}{\bf S}
\]
which proves the result.\end{proof}

Consider the non-homogeneous problem (\ref{April9b}) with $f\in\mathcal{{L}}_{\gamma}^{+}(\Pi)$.
This problem has a solution $w\in{\mathcal H}_{\gamma}^{-}$ which
admits the asymptotic representation 
\begin{equation}
w=\sum_{j=1}^{N}\sum_{\tau=\pm}C_{j\tau}^{1}W_{j}^{\tau}+\sum_{j=1}^{N}\sum_{\tau=\pm}C_{j\tau}^{2}V_{j}^{\tau}+R,\;\; R\in{\mathcal H}_{\gamma}^{+}\label{15April7a}
\end{equation}
which is an rearrangement of the representation (\ref{April9c}) This
motivate the following definition of the space ${\mathcal H}_{\gamma}^{out}$
consisting of vector functions $w\in{\mathcal H}_{\gamma}^{-}$ which
admits the asymptotic representation (\ref{15April7a}) with $C_{j\tau}^{2}=0$.
The norm in this space is defined by 
\begin{equation*}
||w;{\mathcal H}_{\gamma}^{{\rm out}}||=(||R;{\mathcal H}_{\gamma}^{+}||^{2}+\sum_{j=1}^{N}\sum_{\tau=\pm}\,|C_{j\tau}^{1}|^{2})^{1/2}.\label{15April7b}
\end{equation*}
Now, we note that
the kernel in Theorem \ref{tTh1} can be equivalently spanned by

\begin{equation}
Z_{k}^{\tau}=W_{k}^{\tau}+{\bf \tilde{S}}_{k\tau}^{1-}V_{1}^{-}+\sum_{\theta=\pm}\sum_{j=2}^{N}{\bf \tilde{S}}_{k\tau}^{j\theta}V_{j}^{\theta}+\tilde{Z}_{\text{}k}^{\tau},\:\:\:\tilde{Z}_{k}^{\tau}\in{\mathcal H}_{\gamma}^{+},\label{eq:Zk}
\end{equation}
$\tilde{Z}_{k}^{\tau}\in{\mathcal H}_{\gamma}^{+}$, where the incoming
and outgoing waves were interchanged (compare with (\ref{zk})) and
$\tilde{S}$ is a scattering matrix corresponding to that exchange.

\begin{theorem}\label{thrm_iso-1} For any $f\in\mathcal{L}_{\gamma}^{+}(\Pi)$,
problem {\rm(\ref{April9b})} has a unique solution $w\in{\mathcal H}_{\gamma}^{out}$
and the following estimate holds 
\begin{equation}
||w;{\mathcal H}_{\gamma}^{{\rm out}}||\leq c||f;\mathcal{L}_{\gamma}^{+}(\Pi)||,\label{15April7c}
\end{equation}
where the constant $c$ is independent of $\varepsilon\in[0,\varepsilon_{0}]$
and $|\delta|\leq\delta_{0}$. Moreover, 
\begin{equation}
-iC_{j\tau}^{1}=\int_{\Pi}f\cdot\overline{Z_{j}^{\tau}}dxdy.\label{April9d}
\end{equation}
\end{theorem}

\begin{proof} {\em Existence.} According to Theorem \ref{T3s2}
there exists a solution to (\ref{April9b}) of the form (\ref{15April7a}).
Subtracting a linear combination of the elements ${\bf z}_{j}^{\pm}$,
we obtain a solution from ${\mathcal H}_{\gamma}^{out}$.

{\em Uniqueness} is proved in the same way as the isomorphism property
of the mapping ${\mathcal J}$ in Theorem \ref{tTh1}. 

To prove (\ref{April9d}) we multiply equation (\ref{April9b}) by
$\overline{Z_{j}^{\tau}}$and integrate over $\Pi_{R}=(-R,R)\times(0,1)$,
that leads to 
\[
\int_{\Pi_{R}}f\cdot\overline{Z_{j}^{\tau}}dxdy=Q_{R}(w,Z_{j}^{\tau}),
\]
where integration by parts is applied. Now using relation (\ref{eq:Zk})
and (\ref{eq:Qprop}) and sending $R$ to infinity we arrive at (\ref{April9d}).
From representation (\ref{April9d}) if follows that
\begin{equation}
|C_{j\tau}^{1}|\leq C_{\gamma}||f;\mathcal{L}_{\gamma}^{+}(\Pi)||.\label{eq:coeffestimate}
\end{equation}
Now estimate (\ref{15April7c}) follows from (\ref{eq:coeffestimate})
combined with (\ref{15April7a}) where $C_{j\tau}^{2}=0$. From (\ref{15April7a})
the remainder $R\in{\mathcal H}_{\gamma}^{+}$ satisfies 
\begin{equation*}
({\mathcal D}+(\delta\mathcal{P}-\omega) I)R=F,\label{eq:DF}
\end{equation*}
where $F:=f-\delta\mathcal{P}\sum_{j=1}^{N}\sum_{\tau=\pm}\Big(C_{j\tau}^{1}\chi_{\tau}{\bf w}_{j}^{-\tau}\Big)-\sum_{j=1}^{N}\sum_{\tau=\pm}\Big(C_{j\tau}^{1}[\mathcal{D},\chi_{\tau}]{\bf w}_{j}^{-\tau})$.
Therefore, we get
\begin{align}
||F;\mathcal{L}_{\gamma}^{+}(\Pi)|| & \leq||f;\mathcal{L}_{\gamma}^{+}(\Pi)||+\sum_{j=1}^{N}\sum_{\tau=\pm}\Big(C_{j\tau}^{1}||\delta\mathcal{P}\chi_{\tau}{\bf w}_{j}^{-\tau}||\Big)\nonumber \\
 & +\sum_{j=1}^{N}\sum_{\tau=\pm}\Big(C_{j\tau}^{1}||[\mathcal{D},\chi_{\tau}]{\bf w}_{j}^{-\tau}||\Big)\leq c||f;\mathcal{L_{\gamma}}^{+}(\Pi)||,\label{eq:fF}
\end{align}
where the last inequality follows from (\ref{eq:coeffestimate}). Moreover, it follows from
Theorem (\ref{T1a}) that
\begin{equation}
||R;L_{\gamma}^{\pm}(\Pi)||\leq c_{1}||F;L_{\gamma}^{\pm}(\Pi)||.\label{eq:Restimate}
\end{equation}
Combining (\ref{eq:Restimate}) with (\ref{eq:fF}), we get the estimate
\[
||R;\mathcal{H}_{\gamma}^{+}(\Pi)||\leq c||f;\mathcal{L}_{\gamma}^{+}(\Pi)||,
\]
\end{proof}
which together with (\ref{eq:coeffestimate}) leads to (\ref{15April7c}).

\subsection{Analyticity of the scattering matrix}\label{sub:AS}

We represent ${\bf S}$ as 
\begin{equation}
{\bf S}=\Bbb I+{\bf s},\;\;\;\mbox{or, equivalently,}\;\;{\bf S}_{j\tau}^{k\theta}=\delta_{k,j}\delta_{\tau,\theta}+{\bf s}_{j\tau}^{k\theta}.\label{Apr9e}
\end{equation}

\begin{theorem}\label{lemma_analytic-1} The scattering matrix ${\bf S}(\varepsilon,\delta)$
depends analytically on small parameters $\varepsilon\in[0,\varepsilon_{0}]$
and $\delta\in[-\delta_{0},\delta_{0}]$. Moreover, 
\begin{equation}
{\bf s}_{j\tau}^{k\theta}=-i\delta\int_{\Pi}\,\mathcal{P}{\bf w}_{j}^{\tau}\cdot\overline{{\bf w}_{k}^{\theta}}dxdy+O(\delta^{2}).\label{Apr8c}
\end{equation}
\end{theorem}

\begin{proof} Consider the operator 
\[
A_{\gamma}^{out}(\varepsilon,0)\,:\,{\mathcal H}_{\gamma}^{out}\to\mathcal{L}_{\gamma}^{+}(\Pi).
\]
Then it is isomorphism and 
\[
w=\Big(A_{\gamma}^{out}(\varepsilon,0)\Big)^{-1}f=\sum_{j=1}^{N}\sum_{\tau=\pm}C_{j\tau}^{1}W_{j}^{{\rm \tau}}+R=:U+R
\]
is given by 
\[
C_{j\tau}^{1}(\varepsilon)=i\int_{\Pi}f\cdot\overline{{\bf w}_{j}^{\tau}}dxdy,\;\; R=\Big(A_{\gamma}^{+}(\varepsilon,0)\Big)^{-1}(f-({\mathcal D}-\omega)U).
\]
We note that the vector function $({\mathcal D}-\omega)U$ has a compact
support in $\overline{\Pi}$ and is analytic in $\varepsilon$. The coefficients
$C_{j\tau}(\varepsilon)$ depend also analytically on $\varepsilon$. Thus
the inverse operator $\Big(A_{\gamma}^{out}(\varepsilon,0)\Big)^{-1}$
analytically depends on $\varepsilon$.

Let $\zeta_{j}^{\tau}=z_{j}^{\tau}-{\bf w}_{j}^{{\rm \tau}}$. Then
this vector function satisfies 
\[
({\mathcal D}+(\delta\mathcal{P}-\omega) I)\zeta_{j}^{\tau}=-\delta\mathcal{P}{\bf w}_{j}^{{\rm \tau}}\;\;\mbox{in \ensuremath{\Pi\:}and}\;\;\zeta_{j}^{\tau}\in{\mathcal H}_{\gamma}^{out}.
\]
One can check that the solution to this problem is given by the folowing
Neumann series 
\[
\zeta_{j}^{\tau}=\sum_{k=1}^{\infty}(-(A_{\gamma}^{out}(\varepsilon,0))^{-1}\delta P)^{k}{\bf w}_{j}^{\tau},
\]
which represents an analytic function with respect to $\varepsilon$
and $\delta$.

Furthermore, 
\[
({\mathcal D}-\omega)\zeta_{j}^{\tau}=-\delta\mathcal{P}\sum_{k=0}^{\infty}(-(A_{\gamma}^{out}(\varepsilon,0))^{-1}\delta\mathcal{P})^{k}{\bf w}_{j}^{\tau}.
\]
According to (\ref{eq:Qprop}) and (\ref{Apr9e}) 
\[
{\bf s}_{j\tau}^{k\theta}=-i\delta\int_{\Pi}\mathcal{P}\sum_{k=0}^{\infty}(-(A_{\gamma}^{out}(\varepsilon,0))^{-1}\delta\mathcal{P})^{k}{\bf w}_{j}^{{\rm \tau}}\cdot\overline{{\bf w}_{k}^{\theta}}dxdy,
\]
which implies (\ref{Apr8c}). \end{proof}

\section{Trapped modes}\label{sec:TrappedModes}

\subsection{Necessary and sufficient conditions for the existence of trapped modes solutions}\label{sub:NS}

In this section we present a necessary and sufficient condition for existence of a trapped mode in terms of the scattering matrix.

As before we consider problem (\ref{DiracP1}), (\ref{Dirac}) assuming that $|\delta|\leq\delta_{0}$
and $\varepsilon\in(0,\varepsilon_{0}]$.

\begin{theorem}\label{ThTrap} Problem {\rm(\ref{DiracP1})}, {\rm(\ref{Dirac})} has a non-trivial
solution in ${\mathcal H}_{0}$ (a trapped mode), if and only if the matrix
\[
{\bf S}_{\dagger\dagger}+d(\varepsilon)\Upsilon,\;\;\;\Upsilon=\left(\begin{array}{cc}
0 & 1\\
1 & 0
\end{array}\right),
\]
is degenerate. Here $d$ is the quantity defined by {\rm(\ref{Apr9aa})}.

\end{theorem} \begin{proof}
If $w\in{\mathcal H}_{0}$ is a solution
to (\ref{Dirac}) then certainly $w\in\ker A_{-\gamma}$ and hence \footnote{As before {\bf V}+{\bf S}{\bf W}+{\bf r} is a column vector.}
\[
w=a({\bf V}+{\bf S}{\bf W}+{\bf r}),
\]
where $a=(a_{\bullet},a_{\dagger})\in\Bbb C^{2N-2}$ and ${\bf V}$,
${\bf W}$ and ${\bf r}$ are the vector functions from the representation
of the kernel of $A_{-\gamma}$ in (\ref{z_sol}). Using the splitting
of vectors and the scattering matrix in $\bullet$ and $\dagger$
components we write the above relation as 
\begin{equation*}
w=a_{\bullet}({\bf V_{\bullet}}+{\bf S}_{\bullet\bullet}{\bf W}_{\bullet}+{\bf S}_{\bullet\dagger}{\bf W}_{\dagger}+{\bf r}_{\bullet})+a_{\dagger}({\bf V_{\dagger}}+{\bf S}_{\dagger\dagger}{\bf W}_{\dagger}+{\bf S}_{\dagger\bullet}{\bf W}_{\bullet}+{\bf r}_{\dagger}).\label{Apr10a}
\end{equation*}
The first term in the right-hand side contains waves oscillating at
$\pm\infty$ and to guarantee vanishing of this term there we must
require $a_{\bullet}=0$. Since ${\bf r}$ vanish at $\pm\infty$
the requirement $w\in{\mathcal H}_{0}$ is equivalent to the following
demand: 
\begin{equation}
a_{\dagger}({\bf V_{\dagger}}+{\bf S}_{\dagger\dagger}{\bf W}_{\dagger}+{\bf S}_{\dagger\bullet}{\bf W}_{\bullet})\;\;\mbox{vanishes at \ensuremath{\pm\infty}.}\label{Apr10b}
\end{equation}
Using that ${\bf S}$ is unitary and $a_{\bullet}=0$, we get 
\[
|a_{\dagger}|^{2}=|a|^{2}=|{\bf S}a|^{2}=|a_{\dagger}{\bf S}_{\dagger\bullet}|^{2}+|a_{\dagger}{\bf S}_{\dagger\dagger}|^{2}=|a_{\dagger}{\bf S}_{\dagger\bullet}|^{2}+|da_{\dagger}|^{2}.
\]
Since $|d|=1$ this implies $a_{\dagger}{\bf S}_{\dagger\bullet}=0$
and the relation (\ref{Apr10b}) takes the form 
\begin{equation}
a_{\dagger}({\bf V_{\dagger}}+{\bf S}_{\dagger\dagger}{\bf W}_{\dagger})\;\;\mbox{ vanishes at \ensuremath{\pm\infty}.}\label{Apr10c}
\end{equation}
Taking into account representations (\ref{mf2-2}) and (\ref{mf1-1-2})
and equating the coefficients for increasing exponents at $\pm\infty$
we arrive at the relations 
\[
a_{1}\alpha_{2}+(a_{1},a_{2}){\bf S}_{\dagger\dagger}(0,\alpha_{1})^{T}=0,\;\;\; a_{2}\beta_{1}+(a_{1},a_{2}){\bf S}_{\dagger\dagger}(\beta_{2},0)^{T}=0,
\]
where $a_{\dagger}=(a_{1},a_{2})$. Due to the definition of $d$,
this is equivalent to $a_{\dagger}({\bf S}_{\dagger\dagger}+d\Upsilon)=0$
and then expression (\ref{Apr10c}) decays exponentially at $\pm\infty$.
\end{proof}

\subsection{Proof of Theorem {\ref{Texistence}}}\label{sub:PT1}

To prove Theorem \ref{Texistence} it is sufficient to construct a potential $\mathcal{P}$ (subject to certain conditions) which produces a trapped mode. According to Theorem \ref{ThTrap}, we must find a solution to the equation 
\begin{equation}
\det({\bf S}_{\dagger\dagger}+d\Upsilon)=0.\label{Apr10ab}
\end{equation}
To analyse this equation, we write 
\begin{align}
d(\varepsilon) & =e^{i\sigma},\;\;\;{\bf S}=\Bbb I+{\bf s},\:\:\:\mathbf{s}=:-i\delta s,\label{eq:representations}
\end{align}
where $\sigma$ is a real number close to $0$ and as according to (\ref{Apr8c}) ${\bf s}$ is of order $\delta$, then a newly introduced matrix $s$ is of order $1$. To get a relation between $\sigma$
and $\varepsilon$, we can use (\ref{eq:dexpansion}) which gives

\begin{equation}
\cos\sigma =1+O(\varepsilon)\nonumber,\;\;\;
\sin\sigma= \sqrt{\varepsilon}\lambda_{1}\frac{a_{0}}{b_{0}}+O(\varepsilon^{3/2})=:\sqrt{\varepsilon}C_{d}(1+O(\varepsilon)).\label{eq:sigmasincos}
\end{equation}
We will seek for $\mathcal{P}$ and small $\delta>0$ that fulfil the
relations 
\begin{equation}
s_{\dagger\bullet}=0\;\;\;\mbox{and}\;\;\; s_{N-}^{N+}=0.\label{Apr18a}
\end{equation}
Since ${\bf S}{\bf S}^{*}={\bf S}^{*}{\bf S}=\Bbb I$, we have that
$s_{\bullet\dagger}=0$, $s_{N-}^{N+}=0$ and 
\[
|1-i\delta s_{N-}^{N-}|=|1-i\delta s_{N+}^{N+}|=1.
\]
Thus, (\ref{Apr10ab}) becomes 
\[
(1-i\delta s_{N+}^{N+})(1-i\delta s_{N-}^{N-})=d^{2}.
\]
Since the norm of these vectors is $1$ and both of them close to
$1$ this equation is equivalent to 
\begin{equation}
\Im(1-i\delta s_{N+}^{N+})(1-i\delta s_{N-}^{N-})=\Im d^{2}.\label{Apr18b}
\end{equation}
Now to solve this equation, we fix $\delta=\sin\sigma$, that according
to the expansion (\ref{eq:sigmasincos}) gives
$\delta=\sqrt{\varepsilon}C_{d}(1+O(\varepsilon))$ 
and (\ref{Apr18b}) becomes
\begin{equation}
-\Re(s_{N+}^{N+}+s_{N-}^{N-})-\delta\Im(s_{N+}^{N+}s_{N-}^{N-})=2\cos\sigma=2+O(\varepsilon).\label{Apr19a}
\end{equation}

Let us proceed and write equations (\ref{Apr18a}) and (\ref{Apr19a})
as a system, using the following asymptotic formula
\begin{equation}
s_{j\tau}^{k\theta}(\delta\mathcal{P})=\int_{\Pi}\,\overline{{\bf w}_{k}^{\theta}}\mathcal{P}{\bf w}_{j}^{\tau}dxdy+O(\delta).\label{Apr8cy}
\end{equation}
which follows from (\ref{Apr8c}) and (\ref{eq:representations}).
We obtain the system of $4(2N-3)+3$ equations
\begin{equation}
\Re s_{\dagger\bullet}(\delta\mathcal{P})=0,\;\;\Im s_{\dagger\bullet}(\delta\mathcal{P})=0,\label{Apr11aa-1}
\end{equation}
\begin{equation}
\Re s_{N-}^{N+}(\delta\mathcal{P})=0,\;\;\Im s_{N-}^{N+}(\delta\mathcal{P})=0,\label{eq:Apr11bb-1}
\end{equation}
and
\begin{equation}
2\cos\sigma+\Re(s_{N+}^{N+}(\delta\mathcal{P})+s_{N-}^{N-}(\delta\mathcal{P}))+\delta\Im(s_{N+}^{N+}(\delta\mathcal{P})s_{N-}^{N-}(\delta\mathcal{P}))=0.\label{Apr11ab-1}
\end{equation}
To change those equations from vector to scalar notation, we introduce
set of $4(2N-3)+3$ indices:
\begin{align*}
\mathcal{I} & =\Big\{\alpha=(j,\tau,\theta,\Xi):\\
 & j=1;\,\tau=-;\,\theta=\{+,-\};\,\Xi=\{\Re,\,\Im\};\\
 & j=2,\dots,N-1;\,\tau,\theta=\{+,-\};\,\Xi=\{\Re,\,\Im\};\\
 & j=N;\,\tau=+;\,\theta=-;\,\Xi=\{\Re,\,\Im\};\\
 & j=N;\,\tau=+;\,\theta=+;\,\Xi=\Re\Big\}.
\end{align*}
The indices with $j=1,\dots,N-1$ are related to equation (\ref{Apr11aa-1}),
the indices with $j=N,\,\tau=+,\,\theta=-$ correspond to (\ref{eq:Apr11bb-1})
and the last index $(N,+,+,\Re)$ corresponds to (\ref{Apr11ab-1}).

From the number of equations follows that the potential can be chosen
to have the following form:
\begin{align*}
\mathcal{P}(x,y) & =\Phi(x,y)+\sum_{\alpha\in\mathcal{I}}\eta^{\alpha}\Psi^{\alpha}(x,y)\label{Apr10aa}
\end{align*}
where the functions $\Phi$, $\{\Psi^{\alpha}\}_{\alpha\in\mathcal{I}}$
are continuous, real valued with compact support
in $[-R_{0},R_{0}]\times[0,1]$. The functions are assumed to be fixed
and are subject to a set of conditions that is presented later on
in this section. The unknown coefficients can be chosen from the Banach
Fixed Point Theorem. Using indices $\mathcal{I}$ and (\ref{Apr8cy})
we define 
\[
s_{\alpha}:=\Xi s_{j\tau}^{N\theta},\:\:\:\alpha\neq(N,+,+,\Re);\:\:\: s_{\alpha}:=\Re(s_{N+}^{N+}+s_{N-}^{N-})+\delta\Im(s_{N+}^{N+}\cdot s_{N-}^{N-})),\:\:\:\alpha=(N,+,+,\Re)
\]
 and 
\[
\mathbf{\upsilon_{\alpha}}:=\Xi\Big({\bf w}_{j}^{\tau}\cdot\overline{{\bf w}_{N}^{\theta}}\Big),\:\:\:\alpha\neq(N,+,+,\Re);\:\:\:\mathbf{\upsilon_{\alpha}}:=\Big({\bf w}_{N}^{+}\cdot\overline{{\bf w}_{N}^{+}}+{\bf w}_{N}^{-}\cdot\overline{{\bf w}_{N}^{-}}\Big),\:\:\:\alpha=(N,+,+,\Re).
\]
Now we write (\ref{Apr8cy}) as
\begin{align*}
s_{\alpha}\Big(\delta(\Phi+\sum_{\beta\in\mathcal{I}}\eta^{\beta}\Psi^{\beta})\Big) & =\int_{\Pi}\,(\Phi+\sum_{\beta\in\mathcal{I}}\eta^{\beta}\Psi^{\beta})\upsilon_{\alpha}dxdy-\delta\mu_{\alpha}(\delta,\boldsymbol{\eta})=0,\\
 & \alpha\neq(N,+,+,\Re),
\end{align*}
\begin{align*}
s_{\alpha}\Big(\delta(\Phi+\sum_{\beta\in\mathcal{I}}\eta^{\beta}\Psi^{\beta})\Big) & =2(1-\cos\sigma)+\int_{\Pi}\,(\Phi+\sum_{\beta\in\mathcal{I}}\eta^{\beta}\Psi^{\beta})\upsilon_{\alpha}dxdy\\
 & -\delta\mu_{\alpha}(\delta,\boldsymbol{\eta}),\:\:\:\alpha=(N,+,+,\Re)
\end{align*}
Then (\ref{Apr11aa-1}), (\ref{eq:Apr11bb-1}), (\ref{Apr11ab-1}) 
are combined to
\begin{equation}
\mathcal{M}(\delta,\boldsymbol{\eta}):={\bf \Phi}+{\mathcal A}\boldsymbol{\eta}-\delta\boldsymbol{\mu}(\delta,\boldsymbol{\eta})=-2\delta_{(N,+,+,\Re)}^{\alpha},\label{matrixeq}
\end{equation}
with a vector $\mathcal{M}=\{\mathcal{M}_{\alpha}\}_{\alpha\in\mathcal{I}}$,
a vector $\boldsymbol{\Phi}=\{\Phi_{\alpha}\}_{\alpha\in\mathcal{I}}$
with the elements 
\begin{align}
\Phi_{\alpha} & =\int_{\Pi}\,\Phi\upsilon_{\alpha}dxdy,\label{eq:phimultiplicands}
\end{align}
a matrix ${\mathcal A=}\{\mathcal{A}_{\alpha}^{\beta}\}_{\alpha,\beta\in\mathcal{I}}$
with elements given by 
\begin{align*}
\mathcal{A}_{\alpha}^{\beta} & =\int_{\Pi}\,\Psi^{\beta}\upsilon_{\alpha}dxdy,
\end{align*}
a vector $\boldsymbol{\eta}=\{\eta^{\alpha}\}_{\alpha\in\mathcal{I}}$
with real unknown coefficients and a vector $\boldsymbol{\mu}=\{\mathcal{\mu}_{\alpha}\}_{\alpha\in\mathcal{I}}$
that depends on $\delta$ and $\boldsymbol{\eta}$ analytically (analyticity
follows form Theorem \ref{lemma_analytic-1}). 

Now our goal is to solve system (\ref{matrixeq})
with respect to $\boldsymbol{\eta}$. We will reach it in three steps.
First, we eliminate constant $-2$ in the right-hand side in (\ref{matrixeq})
by an appropriate choice of function $\Phi$. Secondly, we choose
functions $\{\Psi^{\alpha}\}_{\alpha\in\mathcal{I}}$ in such a
way that ${\mathcal A}$ is unit and our system becomes nothing but
$\boldsymbol{\eta}=f(\boldsymbol{\eta})$$ $ (with a certian small
function $f$) and by Banach Fixed Point Theorem is solvable. 
The choice of function $\Phi$ is the following
\begin{equation}
\Phi_{\alpha}=0,\,\,\,\alpha\neq(N,+,+,\Re);\:\:\:\Phi_{\alpha}=-2,\,\,\,\alpha=(N,+,+,\Re)\label{eq:phiequations}
\end{equation}
and it is possible due to the following Lemma.

\begin{lemma} Functions
\begin{equation}
\Re{\bf w}_{N}^{\theta}\cdot\overline{{\bf w}_{j}^{\tau}},\;\Im{\bf w}_{N}^{\theta}\cdot\overline{{\bf w}_{j}^{\tau}},\;\Re{\bf w}_{N}^{-}\cdot\overline{{\bf w}_{N}^{+}},\;\Im{\bf w}_{N}^{-}\cdot\overline{{\bf w}_{N}^{+}},\;({\bf w}_{N}^{+}\cdot\overline{{\bf w}_{N}^{+}}+{\bf w}_{N}^{-}\cdot\overline{{\bf w}_{N}^{-}})\label{eq:flinindep}
\end{equation}
with $k=1,2,\dots,N-1$, $\tau,\theta=\pm$ and as before for $k=1$
there is only $\tau=-$ , are linearly independent.
\label{Lemma2}

\end{lemma}

\begin{proof} We first note that functions (\ref{eq:flinindep})
continuously depend on $\varepsilon$, so for the proof of linear independence,
it is enough to consider the limit case, i.e. $\varepsilon=0$. From
(\ref{eq:sol01-1}), (\ref{eq:sol02}), (\ref{LUD11}), (\ref{eq:w+approx}),
(\ref{eq:w-approx}) it follows that 
\[
{\bf w}_{N}^{\theta}\cdot\overline{{\bf w}_{j}^{\tau}}=e^{-iy(1-\lambda_{j}^{\tau})}(a_{1\theta}^{\tau j}(x)+iya_{2}^{\tau j}(x)).
\]
Hence,

\begin{equation}
\Re{\bf w}_{N}^{\theta}\cdot\overline{{\bf w}_{j}^{\tau}}=a_{1\theta}^{\tau j}(x)\cos(y(1-\lambda_{j}^{\tau}))+a_{2}^{\tau j}(x)y\sin(y(1-\lambda_{j}^{\tau})),\label{eq:l1}
\end{equation}
\begin{equation}
\Im{\bf w}_{N}^{\theta}\cdot\overline{{\bf w}_{j}^{\tau}}=a_{2}^{\tau j}(x)y\cos(y(1-\lambda_{j}^{\tau}))-a_{1\theta}^{\tau j}(x)\sin(y(1-\lambda_{j}^{\tau})),\label{eq:l2}
\end{equation}
\begin{equation}
\Re{\bf w}_{N}^{-}\cdot\overline{{\bf w}_{N}^{+}}=b_{1}(x)+b_{2}(x)y^{2},\label{eq:l3}
\end{equation}

\begin{equation}
\Im{\bf w}_{N}^{-}\cdot\overline{{\bf w}_{N}^{+}}=C_1b_{2}(x)y,\label{eq:l4}
\end{equation}
\begin{equation}
{\bf w}_{N}^{+}\cdot\overline{{\bf w}_{N}^{+}}+{\bf w}_{N}^{-}\cdot\overline{{\bf w}_{N}^{-}}=C_2b_2(x)+2b_{1}(x)+2b_{2}(x)y^{2},\label{eq:l5}
\end{equation}
\[
\]
where $a_{1\theta}^{\tau j}$, $a_{2}^{\tau j}$, $b_{1}$ and $b_{2}$, with $\theta=\pm$, are real, non-trivial functions and $C_{1}$ and $C_{2}$ are non-zero constants. Now, as functions $\cos(y(1-\lambda_{j}^{\tau})),\: y\cos(y(1-\lambda_{j}^{\tau})),\:\sin(y(1-\lambda_{j}^{\tau})),\: y\sin(y(1-\lambda_{j}^{\tau})),1\:,\: y,\: y^{2}$
are linearly independent then functions (\ref{eq:l1}), (\ref{eq:l2}), (\ref{eq:l3}), (\ref{eq:l4})
and (\ref{eq:l5}) are linearly independent provided that (i) $a_{1+}^{\tau j}(x)\neq a_{1-}^{\tau j}(x)$,
(ii) (\ref{eq:l3}) and (\ref{eq:l5}) are linearly independent. The
claim in (i) follows from the linear independence of functions ${\bf w}_{N}^{+}$
and ${\bf w}_{N}^{-}$ (see (\ref{eq:w+approx}) and (\ref{eq:w-approx})). Then
(ii) is true if $2\Re{\bf w}_{N}^{-}\cdot\overline{{\bf w}_{N}^{+}}-({\bf w}_{N}^{+}\cdot\overline{{\bf w}_{N}^{+}}+{\bf w}_{N}^{-}\cdot\overline{{\bf w}_{N}^{-}})$
is non zero and that follows from direct calculation 
\begin{align*}
2\Re{\bf w}_{N}^{-}\cdot\overline{{\bf w}_{N}^{+}}-({\bf w}_{N}^{+}\cdot\overline{{\bf w}_{N}^{+}}+{\bf w}_{N}^{-}\cdot\overline{{\bf w}_{N}^{-}}) & =-C_2b_2(x)\\
=\frac{\omega_{N}(9\omega_{N}^{2}+5)}{3(\omega_{N}^{2}-1)}\Big(\cos(2\kappa_{N}x)+\cos(2\kappa_{N}(x-1))-2\Big).
\end{align*}
\end{proof} 

By Lemma \ref{Lemma2} all the multiplicands of
$\Phi$ in (\ref{eq:phimultiplicands}) are linearly independent.
It follows that it is possible to choose $\Phi$ so that (\ref{eq:phiequations})
holds and equations (\ref{matrixeq}) is
\begin{equation}
{\mathcal A}\boldsymbol{\eta}-\delta\boldsymbol{\mu}(\delta,\boldsymbol{\eta})=0.\label{eq:matrixeqlin}
\end{equation}
Now we set matrix ${\mathcal A}$ to be unit, that is its elements fulfil
the conditions
\begin{equation}
\mathcal{A}_{\alpha}^{\beta}=\delta_{\alpha,\beta},\,\,\,\alpha,\beta\in\mathcal{I}.\label{eq:Aunit}
\end{equation}
Again using Lemma \ref{Lemma2}, it is possible to choose functions
$\{\Psi^{\alpha}\}_{\alpha\in\mathcal{I}}$ so that the conditions
(\ref{eq:Aunit}) are fulfilled and (\ref{eq:matrixeqlin}) reads
\begin{equation}
\boldsymbol{\eta}=\delta\boldsymbol{\mu}(\delta,\boldsymbol{\eta}).\label{eq:contraction}
\end{equation}
Now, as $\delta$ is small, the operator on the right hand side of
equation (\ref{eq:contraction}) is a contraction operator, moreover
$\boldsymbol{\mu}$ is analytic in $\delta$ and $\boldsymbol{\eta}$
so Banach Fixed Point Theorem assures that equation (\ref{eq:contraction})
is solvable for $\boldsymbol{\eta}$.

A numerical example of a potential (leading therm $\Phi$) that produces a trapped mode is
$$\mathcal{P}(x,y)\approx\Phi(x,y)=e^{-(\frac{x-0.5}{0.2})^2}\Big(0.4512e^{-(y+65.6273)^2}-e^{-y^2}+0.4512e^{-(y-65.6273)^2} \Big). $$
and is sketched in Figure  \ref{fig:NR_sketch}.

\begin{figure}[h]
\centering{}\includegraphics[scale=1.0]{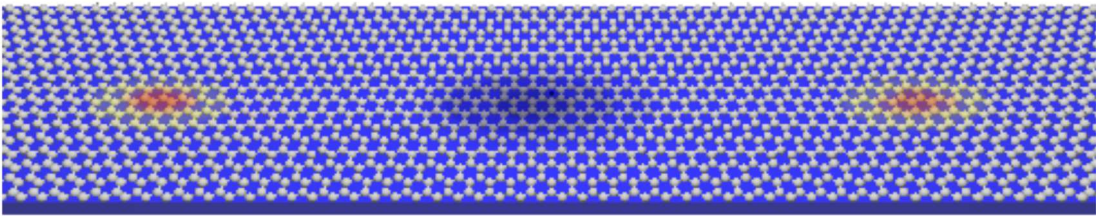}\caption{Sketch of an example potential $\mathcal{P}$ producing a trapped mode with energy close to $\omega_2$. Nanoribbon lies along the y-axis.\label{fig:NR_sketch}}
\end{figure}

\subsection{Proof of Theorem \ref{Tmultiplicity}}\label{sec:mult}

This section is devoted to the proof of the second main result formulated
in Theorem \ref{Tmultiplicity} in the introduction. It concerns the multiplicity of trapped modes and states that (i) there are no trapped modes solutions for energies slightly larger
than threshold, (ii) multiplicity of trapped modes with energies slightly
larger than threshold does not exceed 1 and (iii) the spectrum far
from the threshold is free of trapped modes.

Consider problem (\ref{DiracP1}), (\ref{Dirac}). As previously, ${\mathcal P}$ is a continuous potential with compact support and subject to (\ref{eq:supP}). 

\begin{proof}
(i) Assume the contrary: there exist a trapped mode
solution, i.e. a solution $(u,v)$ of (\ref{Dirac}) belonging to $\mathcal{H}^{1}=\left\{ (u,v):u\in X_{0},v\in Y_{0}\right\} $
( $X_{0}=\{u\in L^{2}(\Pi)\,:\,(i\partial_{x}+\partial_{y})u\in L^{2}(\Pi),\;\; u(0,y)=0\}$,
$Y_{0}=\{v\in L^{2}(\Pi)\,:\,(-i\partial_{x}+\partial_{y})v\in L^{2}(\Pi),\;\; v(1,y)=0\}$).
Now, from Theorem \ref{T3s2} and Proposition \ref{Pr1} (iii), which
assets that all $\lambda$ in the strip $\{|Im\lambda|\leq\gamma\}$, in the exponential part of the solutions $w=(u,v)=e^{-i\lambda y}(\mathcal{U},\mathcal{V})$ are real, it follows that $(u,v)\in\mathcal{H}_{\gamma}^{+}$. Moreover, from Theorem \ref{T1a} operator ${\mathcal D}+(\delta\mathcal{P}-\omega) I\;:\; X_{\gamma}^{\pm}\times Y_{\gamma}^{\pm}\;\to\; L_{\gamma}^{\pm}(\Pi)$
is an isomorphism, so the only solution to $({\mathcal D}+(\delta\mathcal{P}-\omega)I)w=0$
is $w=0$.\\
(ii) There exist at least one trapped mode and it can be constructed through conditions given in the previous section, Sect. \ref{sub:PT1}. Assume now that there
are two trapped modes. According to Theorem \ref{T3s2} and Proposition
\ref{Pr1} (iii), which assets that there are exactly two solutions
$w=(u,v)=e^{-i\lambda y}(\mathcal{U},\mathcal{V})$ with complex $\lambda$ in the strip $\{|\Im\lambda|\leq\gamma\}$,
it follows that the trapped mode is of the form
\begin{equation}
w_{j}=C_{j}e^{-i\lambda_{-}y}(\mathcal{U}_{N}^{-}(x),\mathcal{V}_{N}^{-}(x))+D_{j}e^{-i\lambda_{+}y}(\mathcal{U}_{N}^{+}(x),\mathcal{V}_{N}^{+}(x))+R_{j}\label{eq:2trappedmodes}
\end{equation}
with $R_{j}\in{\mathcal H}_{\gamma}^{+}$ , $j=1,2$. Now, consider the
following linear combination of trapped modes (\ref{eq:2trappedmodes})
\begin{equation}
w_{3}=w_{1}-\frac{C_{1}}{C_{2}}w_{2}=(D_{1}-\frac{C_{1}}{C_{2}}D_{2})e^{-i\lambda_{+}y}(\mathcal{U}_{N}^{+}(x),\mathcal{V}_{N}^{+}(x))+(R_{1}-\frac{C_{1}}{C_{2}}R_{2}),\:\:\: w_{3}\in X_{\gamma}^{-}\times Y_{\gamma}^{-}\label{eq:w3}
\end{equation}
which is a solution to problem (\ref{Dirac})
as well. From Theorem \ref{T1a}, it follows that operator ${\mathcal D}+(\delta\mathcal{P}-\omega) I\;:\; X_{\gamma}^{-}\times Y_{\gamma}^{-}\;\to\; L_{\gamma}^{-}(\Pi)$
is an isomorphism and hence $w_{3}=0$. From (\ref{eq:w3}) we get $w_{1}=Cw_{2}$.\\
(iii) First we choose $\gamma$ such that the strip $|\Im\lambda|\leq\gamma$ contains only real roots of (\ref{K1})
for all $\omega$ described in the Proposition (iii). then we note that supremum with respect to such $\omega$ of the quantity $sup_{\Im \lambda=\pm\gamma}||({\mathcal D}-\omega)^{-1}||_{L^2(\Pi)\rightarrow L^2(\Pi)}$ is bounded then reasoning as in (i), we obtain the estimate for $\delta_1$.

\end{proof}

\section*{Acknowledgement}
The authors thank I. V. Zozoulenko for the discussion on physical aspects of the paper. S. A. Nazarov acknowledges financial support from The Russian Science Foundation (Grant 14-29-00199).

\appendix
\begin{center}
  \Huge\bfseries\appendixpagename
\end{center}
\section{Ellipticity\label{sec:Appendix:-Ellipticity}}

According to the general ellipticity theory \cite{ADN}, it is necessary
to check few simple properties. To this end, we write three different
tables of the ADN-indices:

\begin{center}
$T^{0}:$%
\begin{tabular}{c|c|c|}
\multicolumn{1}{c}{} & \multicolumn{1}{c}{1} & \multicolumn{1}{c}{1}\tabularnewline
\cline{2-3} 
0 & 1 & 1\tabularnewline
\cline{2-3} 
0 & 1 & 1\tabularnewline
\cline{2-3} 
\end{tabular}  $\:\:\:\:\: T^{v}:$%
\begin{tabular}{c|c|c|}
\multicolumn{1}{c}{} & \multicolumn{1}{c}{0} & \multicolumn{1}{c}{1}\tabularnewline
\cline{2-3} 
0 & 0 & 1\tabularnewline
\cline{2-3} 
1 & 1 & 2\tabularnewline
\cline{2-3} 
\end{tabular} $\:\:\:\:\: T^{u}:$%
\begin{tabular}{c|c|c|}
\multicolumn{1}{c}{} & \multicolumn{1}{c}{1} & \multicolumn{1}{c}{0}\tabularnewline
\cline{2-3} 
1 & 2 & 1\tabularnewline
\cline{2-3} 
0 & 1 & 0\tabularnewline
\cline{2-3} 
\end{tabular}
\par\end{center}

Notice that the numbers inside the tables are obtained as the sum
of number standing at the corresponding rows and columns outside the
tables. They indicate orders of differential operators composing the
principal part of the system (\ref{homoDirac2})
\begin{equation}
\mathcal{D}_{0}:=\mathcal{D},\:\:\:\: \mathcal{D}_{v}:=\mathcal{D}+\left(\begin{array}{cc}
-\omega & 0\\
0 & 0
\end{array}\right),\:\:\:\: \mathcal{D}_{u}:=\mathcal{D}+\left(\begin{array}{cc}
0 & 0\\
0 & -\omega
\end{array}\right)\label{eq:opMatrix}
\end{equation}
where $\mathcal{D}_{\alpha}=\mathcal{D}_{\alpha}(\partial_{x},\partial_{y})$ for $\alpha=0,v,u$.
We have $\det \mathcal{D}_{\alpha}(-i\eta,-i\xi)=|\eta|^{2}+|\xi|^{2}$ and,
hence, the operator matrix (\ref{eq:opMatrix}) is elliptic with $\alpha=0,u,v$.
However, it is also necessary to verify the Shapiro-Lopatinskii condition
at the both sides of the strip $\Pi$. For example, for the right edge
$0\times\mathbb{R}$ of the nanoribbon the Cauchy problem
\[
\mathcal{D}_{\alpha}(\partial_{x},-i\xi)\left(\begin{array}{c}
u\\
v
\end{array}\right)=0\:\:\:\mbox{in }\mathbb{R}_{+},\:\:\: u(0,\xi)=1
\]
must have only one solution decaying as $y\rightarrow\infty$.

If $\alpha=0$, we have $u(x,\xi)=e^{\xi x}$ without decay for $\xi>0$.
In the case $\alpha=v$ the general solution takes the form
\[
v(x,\xi)=Ce^{-|\xi|y},\:\:\: u(x,\xi)=C\omega^{-1}(-|\xi|+\xi)e^{-|\xi|y}.
\]
But again the Cauchy problem has no solution for $\xi>0$. Finally,
fixing $\alpha=u$ we obtain the desired solution 
\[
v(x,\xi)=e^{-|\xi|y},\:\:\: u(x,\xi)=i\omega^{-1}(|\xi|+\xi)e^{-|\xi|y}
\]
for any $\xi\in\mathbb{R}\backslash\{0\}$.

A similar calculation shows that the Cauchy problem 
\[
\mathcal{D}_{\alpha}(\partial_{x},-i\xi)\left(\begin{array}{c}
u\\
v
\end{array}\right)=0\:\:\:\mbox{in }\mathbb{R}_{-},\:\:\: v(0,\xi)=1
\]
serving for the left edge of the nanoribbon gets the necessary property
for the case $\alpha=v$ only.

Reviewing the situation, we see that any of three ADN-tables is suit
inside $\Pi$ but none serves simultaneously at both sides of the
nanoribbon. This means that our problem is not included into the standard
elliptic theory.

It also should be mentioned that, if there exists an ADN-table fitting
everywhere in $\Pi$ and on $\partial\Pi$, then according to \cite[Ch.\ 5]{NaPl}, the numbers of incoming and outgoing waves must coincide
in each outlet to infinity. The latter, as we have verified in (\ref{W1a}) and
(\ref{W1b}) is not true.

\section{Mandelstam radiation condition\label{sec:Appendix:-Mandelstam-radiation}}

Here we want to clarify the division of waves in two classes outgoing/incoming
accoriding to the appearance of the $\pm i$ in (\ref{eq:bioosc}). To do so, we
employ the Mandelstam radiation conditions which define classification into outgoing and incoming waves by the direction of the energy transfer
\cite{Mandelstam,Poynting,Umov}.

Let us write initial system (\ref{homoDirac2}) in the form
\begin{align}
(-i\partial_{x}+\partial_{y} & )\mathbf{v}=i\partial_{t}\mathbf{u},\nonumber \\
(-i\partial_{x}-\partial_{y} & )\mathbf{u}=i\partial_{t}\mathbf{v},\label{eq:transsystem}
\end{align}
with
\begin{equation}
\mathbf{u}=e^{-i\omega t}u\:,\:\:\mathbf{v}=e^{-i\omega t}v\:,\:\:\mathbf{w}=(\mathbf{u},\mathbf{v}).\label{eq:transfun}
\end{equation}
Energy transfer from area $\Omega$ is defined as
\[
-\frac{d}{dt}\int_{\Omega}|\partial_{t}\mathbf{w}|^{2}dxdy.
\]
Using relations (\ref{eq:transsystem}), (\ref{eq:transfun}) and
performing partial integration we get 

\begin{align*}
-\frac{d}{dt}\int_{\Omega}|\partial_{t}\mathbf{w}|^{2}dxdy & =-|\omega|^{2}\int_{\Omega}\mathbf{\overline{\mathbf{w}}}\partial_{t}\mathbf{w}+\mathbf{w}\partial_{t}\overline{\mathbf{w}}dxdy\\
 & =-i|\omega|^{2}\int_{\Gamma}\Big(i(\mathbf{u\overline{\mathbf{v}}+\mathbf{v}\overline{\mathbf{u}}}),\mathbf{u}\overline{\mathbf{v}}-\mathbf{v}\overline{\mathbf{u}}\Big)\cdot(n_{x},n_{y})ds,
\end{align*}
where $\Gamma$ is the boundary of the domain $\Omega$. Consider energy transfer along the nanoribbon (along the y-axis) from
$-\infty$ to $+\infty$, that is choose $(n_{x},n_{y})=(0,1)$, then
the last formula is equal to 

\[
-\frac{d}{dt}\int_{\Omega}|\partial_{t}\mathbf{w}|^{2}dxdy=-i|\omega|^{2}\int_{0}^{1}(u\overline{v}-v\overline{u})dx=i|\omega|^{2}q(w,w),
\]
where the last equality comes from the the definition of q-form (\ref{eq:qform}).
Accordingly the energy transfer along the nanoribbon is proportional
to $iq$, which is $\pm1$ for $q=\mp i$. It follows that the value
of q-form defines direction of wave propagation, namely $q=-i$ describes
waves propagating from from $-\infty$ to $+\infty$ and $q=+i$ those
from $+\infty$ to $-\infty$. This leads to the definition of outgoing/incoming
waves (\ref{W1a}), (\ref{W1b}) as those traveling to $\pm\infty$
and from $\pm\infty$ . 

\section{Proof of Theorem \ref{T1} \label{sec:Appendix:Proof1}}

\begin{proof} We prove Theorem for the sign ``+''
in (\ref{2aa}), the proof for the sign ``-'' is the same if we
put $-\sigma$ instead of $\sigma$ in the sequel.

Using the Fourier transform with respect to $y$ 
\[
\hat{g}(\lambda)=\int_{-\infty}^{\infty}g(y)e^{i\lambda y}dy,
\]
with $\lambda=\xi+i\sigma$, we transform problem (\ref{nonH1}), (\ref{nonH2}),
(\ref{2a}) 
\begin{eqnarray}
& & (-i\partial_{x}-i\lambda)\hat{v}-\omega\hat{u}=\hat{g},\label{eq:afterFT0}\\
& & (-i\partial_{x}+i\lambda)\hat{u}-\omega\hat{v}=\hat{h}\;\;\;\mbox{in \ensuremath{(0,1)}}\label{eq:afterFT}
\end{eqnarray}
with the boundary conditions $\hat{u}(0,\lambda)=0$ and $\hat{v}(1,\lambda)=0$.
If $\lambda$ satisfies the condition of the theorem then this problem
has a unique solution for every $\xi\in\Bbb R$. Let us obtain estimates
for $u$ and $v$.

We begin with the case $h=0$. Then $\hat{v}=\omega^{-1}(-i\partial_{x}+i\lambda)\hat{u}$
and 
\begin{equation}
(-\partial_{x}^{2}+\lambda^{2}-\omega^{2})\hat{u}=\hat{g}\omega\;\;\mbox{on \ensuremath{(0,1)}}\label{5}
\end{equation}
with boundary conditions 
\begin{equation}
\hat{u}(0,\lambda)=0,\;\;\partial_{x}\hat{u}(1,\lambda)=\lambda\hat{u}(1,\lambda).\label{6}
\end{equation}
Consider the case when $|\xi|$ is large. We are looking for a solution
to the above problem in the form $\hat{u}=w+R$, where $w$ solves
(\ref{5}) with the Dirichlet boundary condition: $w(0,\lambda)=w(1,\lambda)=0$
and 
\begin{equation*}
R(x,\lambda)=\frac{\partial_{x}w(1,\lambda)\sin\kappa x}{\kappa\cos\kappa-\lambda\sin\kappa},\label{6c}
\end{equation*}
where $\kappa^{2}=\omega^{2}-\lambda^{2}$. Direct calculations show
that 
\[
w(x,\lambda)=\omega\int_{0}^{1}L(x,z)\hat{g}(z,\lambda)dz,
\]
where 
\[
L(x,z)=\frac{1}{\kappa\sin\kappa}\begin{cases}
\sin\kappa x\sin\kappa(1-z)\;\;\mbox{if \ensuremath{z>x}}\\
\sin\kappa(1-x)\sin\kappa z\;\;\mbox{if \ensuremath{x>z}}
\end{cases}.
\]
Since $\kappa=i\tau$, $\tau=\lambda-\omega^{2}\lambda^{-1}/2+O(|\lambda|^{-3})$
and $\lambda=\xi+i\sigma$, therefore, 
\[
|L(x,z)|\leq\frac{C}{|\xi|}e^{-|\xi|\,|x-z|}.
\]
This implies the estimate 
\begin{equation}
\int_{0}^{1}(|\xi|^{4}|w|^{2}+|\xi|^{2}|\partial_{x}w|^{2}+|\partial_{x}^{2}w|^{2})dx\leq C\int_{0}^{1}|\hat{g}|^{2}dz.\label{7d}
\end{equation}
As a consequence, we get 
\begin{equation}
|\partial_{x}w(1,\lambda)|^{2}\leq C|\xi|^{-1}\int_{0}^{1}|\hat{g}|^{2}dz.\label{7dd}
\end{equation}
Now, from trigonometric function properties 
\begin{eqnarray}
 &  & \kappa\cos\kappa x-\lambda\sin\kappa x=i\tau\cosh\tau x-i\lambda\sinh\tau x\nonumber \\
 &  & =i\Big(-\frac{\omega^{2}}{4\lambda}+O(|\lambda|^{-3})\Big)e^{\tau x}+i\Big(\lambda+O(|\lambda|^{-1})\Big)e^{-\tau x}.\label{eq:cosh}
\end{eqnarray}
Therefore, 
\begin{align*}
\int_{0}^{1}|(\partial_{x}-\lambda)R|^{2}dx & =\int_{0}^{1}|\kappa\cos\kappa x-\lambda\sin\kappa x|^{2}dx\\
 & \leq C\frac{|\partial_{x}w(1,\lambda)|^{2}}{|\kappa\cos\kappa-\lambda\sin\kappa|^{2}}\Big(\frac{|e^{2\xi}-1|}{|\xi|^{3}}+|e^{-2\xi}-1|\,|\xi|\Big)
\end{align*}
and using again (\ref{eq:cosh}) with $x=1$, we get 
\[
\int_{0}^{1}|(\partial_{x}-\lambda)R|^{2}dx\leq C\frac{|\partial_{x}w(1,\lambda)|^{2}}{|\xi|}
\]
and 
\[
\int_{0}^{1}|R|^{2}dx\leq C|\xi|\,|\partial_{x}w(1,\lambda)|^{2}.
\]
The last two estimates together with (\ref{7d}) and (\ref{7dd})
give 
\begin{equation}
\int_{0}^{1}(|\hat{u}|^{2}+|\hat{v}|^{2})dx\leq C\int_{0}^{1}|\hat{g}|^{2}dx\label{7ddd}
\end{equation}
for large $\xi$. Estimate (\ref{7ddd}) for $\xi$ from a certain
bounded interval can be obtained directly by analysing problem (\ref{5}),
(\ref{6}), since it is elliptic and generate an isomorphic operator
due to the assumption on $\lambda$. Thus estimate (\ref{7ddd}) is
valid for all real $\xi$. Using (\ref{eq:afterFT0}), (\ref{eq:afterFT}) we can estimate
also $L^{2}$-norms of $(-i\partial_{x}+i\lambda)\hat{v}$ and $(-i\partial_{x}-i\lambda)\hat{u}$.
Now reference to Parseval's theorem gives 
\[
||u;L_{\sigma}^{+}(\Pi)||+||(i\partial_{x}+\partial_{y})u;L_{\sigma}^{+}(\Pi)||+||v;L_{\sigma}^{+}(\Pi)||+||(-i\partial_{x}+\partial_{y})v;L_{\sigma}^{+}(\Pi)||\leq C||g;L_{\sigma}^{+}(\Pi)||
\]
in the case $h=0$. The change of variables $x\to1-x$ and $(u,v)\to(-v,u)$
reduces the case $g=0$ to the previous one. The theorem is proved.\end{proof}

\section{Proof of Theorem \ref{T1s} \label{sec:Appendix:Proof2}}

\begin{proof} 

From Theorem \ref{T1}, we know that $(u^{+},v^{+})$can
be expressed as follows

\[
(u^{+},v^{+})^T=\frac{1}{2\pi}\int_{\Im\lambda=\gamma}e^{-i\lambda y}(\mathcal{D}(\partial_{x},i\lambda)-\omega I)^{-1}(\hat{g},\hat{h})^Td\lambda
\]
where $D=D(\partial_{x},\partial_{y})$ was defined in (\ref{homoDirac2}).
Now choose a positive value $\rho$ sufficiently large so that all eigenvalues
$\lambda_{j}^{\tau}$ described in Proposition \ref{Pr1} (iii) are
contained in the set $\{\lambda\in\mathbb{C}:|\Im\lambda|<\gamma,|\Re\lambda|<\rho\}$
(Figure \ref{fig:Plot_int}). This is possible according to Proposition
\ref{Pr1}. Applying Cauchy's formula, we get
\[
(u^{+},v^{+})^T=\frac{1}{2\pi}\int_{-\rho+i\gamma}^{\rho+i\gamma}e^{i\lambda y}(\mathcal{D}(\partial_{x},i\lambda)-\omega I)^{-1}(\hat{g},\hat{h})^Td\lambda
\]
\begin{equation}
=\frac{1}{2\pi}\Bigl(\int_{-\rho-i\gamma}^{\rho-i\gamma}\dots d\lambda+\int_{\rho-i\gamma}^{\rho+i\gamma}\dots d\lambda-\int_{-\rho-i\gamma}^{-\rho+i\gamma}\dots d\lambda\Bigr)\label{eq:3integrals}
\end{equation}
\[
+i\sum_{\tau=\pm}\sum_{j=1}^{N}\mbox{Res}\Bigl(e^{-i\lambda y}(\mathcal{D}(\partial_{x},i\lambda)-\omega I)^{-1}(\hat{g},\hat{h})^T\Bigr)|_{\lambda=\lambda_{j}^{\tau}}
\]

\begin{figure}[h]
\centering{}\includegraphics[scale=0.5]{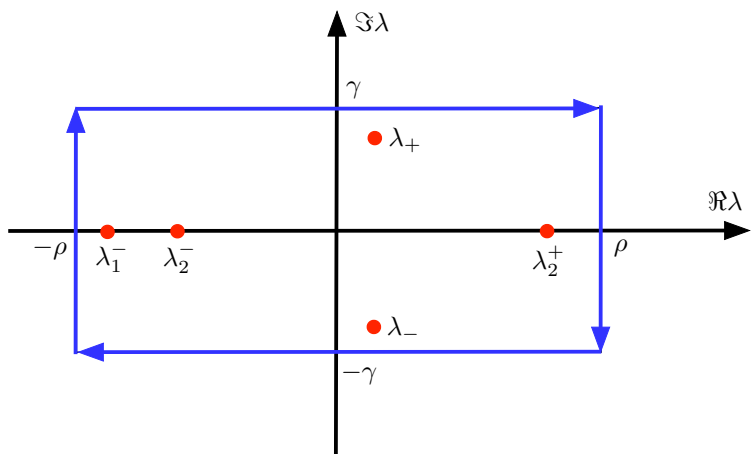}\caption{Schematic figure showing a region that contains all eigenvalues
$\lambda_{j}^{\tau}$ (red dots) described in Proposition (\ref{Pr1}
iii) with blue contour used in the integration (\ref{eq:3integrals})\label{fig:Plot_int}}
\end{figure}
The first integral on the right hand side tends
to $(u^{-},v^{-})$with $\rho\rightarrow\infty$. Moreover the
last two integrals tend to zero for smooth functions $(g,h)$ with
compact support. It is enough to prove the theorem for such functions
as they are dense in $L_{\gamma}^{+}(\Pi)\cap L_{\gamma}^{-}(\Pi)$.

The residua in (\ref{eq:3integrals}) belong to
the kernel of $(\mathcal{D}(\partial_{x},\partial_{y})-\omega  I)$. Therefore
the last sum is linear combination of solutions ${\bf w}_{k}^{\tau}$,
with $\tau=\pm$ and $k=1,\dots,N$ and we obtain (\ref{Apr8a}) with
certian coefficinets. Now, we want to find expressions for those coefficients.

Let us
define a smooth function $\eta_{-}=\eta_{-}(y)$ such that $\eta_{-}(y)=1$
in the neighbourhood of $-\infty$ and $\eta_{-}(y)=0$ in the neighbourhood
of $+\infty$. Using the biortogonality conditions for functions ${\bf w}_{j}^{\mp}$ in (\ref{eq:bioosc})
we get
\begin{equation}
\int_{\Pi}\overline{{\bf w}_{j}^{+}}({\mathcal D}-\omega I)(\eta_{-}(u^{+},v^{+})^T-\eta_{-}(u^{-},v^{-})^T)dxdy=-iC_{j}^{+}.\label{Apr8a-1}
\end{equation}
Now note that 
\[
\eta_{-}(u^{-},v^{-}),(1-\eta_{-})(u^{+},v^{+})\in L_{\gamma}^{+}(\Pi)\cap L_{\gamma}^{-}(\Pi).
\]
Applying integration by parts, follows
\[
\int_{\Pi}\overline{{\bf w}_{j}^{+}}({\mathcal D}-\omega I)((1-\eta_{-})(u^{+},v^{+})^T)dxdy=0
\]
and
\[
\int_{\Pi}\overline{{\bf w}_{j}^{+}}({\mathcal D}-\omega I)(\eta_{-}(u^{-},v^{-})^T)dxdy=0
\]
so from (\ref{Apr8a-1}), we get 
\[
\int_{\Pi}\overline{{\bf w}_{j}^{+}}\cdot(g,h)dxdy=-iC_{j}^{+}.
\]
In a similar way we obtain 
\[
\int_{\Pi}\overline{{\bf w}_{j}^{-}}\cdot(g,h)dxdy=iC_{j}^{-}.
\]
This furnishes the assertion.
\end{proof}

\end{document}